\documentclass[UKenglish,cleveref,autoref]{lmcs}
\pdfoutput=1
\usepackage[utf8]{inputenc}

\usepackage{lastpage}
\lmcsdoi{18}{1}{5}
\lmcsheading{}{\pageref{LastPage}}{}{}%
{Nov.~26,~2020}{Jan.~11,~2022}{}

\keywords{Integrity constraints, the  implication problem}

\usepackage{tabularx}
\newcolumntype{Y}{>{\centering\arraybackslash}X}

\usepackage{amsfonts}
\usepackage{booktabs}
\usepackage{siunitx}
\usepackage{multirow}
\usepackage{graphicx}
\usepackage{listings}
\usepackage{commath}
\usepackage{epstopdf}
\usepackage{xcolor}
\usepackage{url}
\usepackage{mathtools}
\usepackage{caption}
\usepackage{alltt}
\usepackage{mathrsfs}
\usepackage{float}
\usepackage{caption}
\usepackage[normalem]{ulem}
\usepackage{graphicx,fancyvrb}
\usepackage{todonotes}
\usepackage[linesnumbered,ruled,vlined]{algorithm2e}
\usepackage{algpseudocode}
\usepackage{amsmath,amssymb}
\usepackage{booktabs}
\usepackage{caption}
\usepackage{float}
\usepackage{capt-of}
\usepackage{multirow}

\usepackage{array}
\usepackage{arydshln}
\usepackage{todonotes}

\definecolor{mygreen}{rgb}{0,0.6,0}
\definecolor{mygray}{rgb}{0.5,0.5,0.5}
\definecolor{mymauve}{rgb}{0.58,0,0.82}
\definecolor{cadmiumgreen}{rgb}{0.0, 0.42, 0.24}

\usepackage{listings}
\lstnewenvironment{VerbatimText}[1][]{
    
    \lstset{fancyvrb=true,basicstyle=\footnotesize,captionpos=b,xleftmargin=2em,#1}
}{}

\newcommand{\pr}{{\tt \mathrm{Pr}}}

\newcommand{\ignore}[1]{}

\newcommand{\dan}[1]{{\texttt{\color{red} Dan: [{#1}]}}}

\newcommand{\added}[1]{{\color{black}{#1}}}

\newcommand{\sx}{\mathbf{x}}
\newcommand{\sy}{\mathbf{y}}
\newcommand{\sz}{\mathbf{z}}
\newcommand{\su}{\mathbf{u}}
\newcommand{\sv}{\mathbf{v}}
\newcommand{\se}{\mathbf{e}}
\newcommand{\bs}{\mathbf{s}}

\newcommand\mvd{\twoheadrightarrow}

\newcommand{\fd}{\rightarrow}

\newcommand{\real}{{\mathbb{R}}}
\newcommand{\realPos}{{\mathbb{R}_+}}

\newcommand*{\rom}[1]{\expandafter\@slowromancap\romannumeral #1@}
\newcommand{\RNum}[1]{\uppercase\expandafter{\romannumeral #1\relax}}

\newtheorem{notation}[thm]{Notation}

\newcommand{\sel}[1]{{\sigma}}

\newcommand{\cut}[1]{}
\newcommand{\eat}[1]{}

\newcommand{\defeq}{\stackrel{\text{def}}{=}}
\newcommand{\setof}[2]{\{{#1}\mid{#2}\}}        

\newcommand{\R}{{\mathbb R}}

\def\closure{\mathbf{cl}}
\def\conhull{\mathbf{conhull}}
\def\set#1{\mathord{\{#1\}}}
\def\conehull#1{\mathbf{conhull}\left(#1\right)}
\def\cl#1{\mathbf{cl}\left(#1\right)}

\def\iset{\mathrm{m}}
\def\isetc{\mathrm{m}^c}

\def\measure{\mu}
\def\imeasure{\mu^*}
\def\eqdef{\mathrel{\stackrel{\textsf{\tiny def}}{=}}}

\def\I{\mathcal{I}}

\def\D{\mathcal{D}}
\def\e#1{\emph{#1}}

\def\implies{\Rightarrow}

\newcommand{\pow}[1]{2^{{#1}}} 

\def\vplus{{+}}

\newenvironment{repeatresult}[2]
{\vskip0.5em\par\textsc{#1} #2.\em}
{\vskip1em}

\newenvironment{reptheorem}[1]{\begin{repeatresult}{Theorem}{#1}}{\end{repeatresult}}

\def\appendix{\par
	\section*{APPENDIX}
	\setcounter{section}{0}
	\setcounter{subsection}{0}
	\def\thesection{\Alph{section}} }

\def\entropicFunctions{\Gamma^*}
\def\entropicPlhdrl{\Gamma}

\def\positiveConen{\mathcal{P}_n}

\def\stepfn{\mathcal{S}_n}

\def\iset{\mathrm{m}}
\def\isetc{\mathrm{m}^c}

\def\measure{\mu}
\def\imeasure{\mu^*}
\def\eqdef{\mathrel{\stackrel{\textsf{\tiny def}}{=}}}
\def\I{\mathcal{I}}

\def\B{\mathcal{B}}
\def\e#1{\emph{#1}}

\def\vneg{\mathbf{neg}}

\def\dotprod{{\cdot}}
\def\impliedCI{\tau}

\def\APXI{{\tt \textsc{ApxI}}}
\def\EI{\tt \textsc{EI}}

\begin{document}
\title[Integrity Constraints Revisited]{Integrity Constraints Revisited:\texorpdfstring{\\}{} From Exact to Approximate Implication}

\author[B.~Kenig]{Batya Kenig\rsuper{a}}	

\author[D.~Suciu]{Dan Suciu\rsuper{b}}	

\address{Technion}	
\email{batyak@technion.ac.il}  

\address{University of Washington}	
\email{suciu@cs.washington.edu}  


\begin{abstract}
	Integrity constraints such as functional dependencies (FD) and
multi-valued dependencies (MVD) are fundamental in database schema
design.  Likewise, probabilistic conditional independences (CI) are
crucial for reasoning about multivariate probability distributions.
The implication problem studies whether a set of constraints
(antecedents) implies another constraint (consequent), and has been
investigated in both the database and the AI literature, under the
assumption that all constraints hold {\em exactly}.  However, many
applications today consider constraints that hold only {\em
approximately}.  In this paper we define an approximate implication as
a linear inequality between the degree of satisfaction of the
antecedents and consequent, and we study the {\em relaxation
problem}: when does an exact implication relax to an approximate
implication?  We use information theory to define the degree of
satisfaction, and prove several results.  First, we show that any
implication from a set of data dependencies (MVDs+FDs)
can be relaxed to a simple linear inequality with a factor at most
quadratic in the number of variables; when the consequent is an FD,
the factor can be reduced to 1.
Second, we prove that there exists an
implication between CIs that does not admit any relaxation; however,
we prove that every implication between CIs relaxes ``in the limit''.
Then, we show that the implication problem for
differential constraints in market basket analysis also admits a
relaxation with a factor equal to 1. Finally, we show how some of the results in the paper can be derived using the {\em I-measure} theory, which relates between information theoretic measures and set theory.
Our results recover, and
sometimes extend, several previously known results about the
implication problem: the implication of MVDs and the implication of
differential constraints for frequent item sets can be checked by
considering only 2-tuple relations.

\end{abstract}
\maketitle

\section{Introduction}

\label{sec:intro}

\begin{table*}[t]
  \begin{tabularx}{\linewidth}{cYYYc}
			\toprule
			\multirow{3}{*}{Cone}&\multicolumn{4}{c}{Relaxation Bounds}\\
			\cmidrule{2-5}
			&\multirow{2}{*}{General} &{MVDs+FDs}&{MVDs+FDs}&Disjoint MVDs+FDs\\		
			&&$\implies$ FD&$\implies$ any&$\implies$ MVD/FD\\
			\midrule
			$\entropicPlhdrl_n$&$(2^n)!$~(Thm.~\ref{thm:RelaxationInGamman})&$1$~(Thm. \ref{thm:MainInGamman})&$\frac{n^2}{4}$~(Thm. \ref{thm:MainInGamman})&$1$~(Thm. \ref{thm:DisjointSaturated})\\
			$\entropicFunctions_n$&$\infty$~~(Thm. \ref{th:no:approx})&$1$~(Thm. \ref{thm:MainInGamman})&$\frac{n^2}{4}$~(Thm. \ref{thm:MainInGamman})&$1$~(Thm. \ref{thm:DisjointSaturated})\\	
			$\positiveConen$&$1$~(Thm. \ref{thm:PositiveConeUnitRelaxation})&$1\hspace{0.1cm}$(Thm. \ref{thm:PositiveConeUnitRelaxation})&$1$~(Thm. \ref{thm:PositiveConeUnitRelaxation})&$1$~(Thm. \ref{thm:PositiveConeUnitRelaxation})\\
			\bottomrule
		\end{tabularx}
		\caption{Summary of results: relaxation bounds for the implication $\Sigma \implies \tau$ for the sub-cones of $\entropicPlhdrl_n$  under various restrictions.
                (1) \textit{General}; no restrictions to either $\Sigma$ or $\impliedCI$. (2) $\Sigma$ is a set of saturated CIs and conditional entropies (i.e., MVDs+FDs in databases), and  $\impliedCI$ is a conditional entropy. (3) $\Sigma$ is a set of saturated CIs and conditional entropies,  $\impliedCI$ is \e{any} CI. (4) \textit{Disjoint} integrity constraints
                where the terms in $\Sigma$ are both saturated and \e{disjoint} (see Definition~\ref{def:disjoint} in Sec.~\ref{sec:PBoundedRelaxations}), and $\impliedCI$ is saturated.}
		\label{tab:AIResults}
\end{table*}

Traditionally, integrity constraints are assertions about a database that are stated by the database administrator and enforced by the system during updates.  However, in 
several applications of Big Data, integrity constraints are discovered, or mined in a database
instance, as opposed to being asserted by the administrator~\cite{DBLP:journals/is/GiannellaR04,DBLP:conf/vldb/SismanisBHR06,DBLP:conf/icde/ChuIPY14,DBLP:conf/cikm/BleifussBFRW0PN16,DBLP:journals/pvldb/0001N18}.
For example, data cleaning can be done by first learning conditional
functional dependencies in some reference data, then using them to
identify inconsistencies in the test
data~\cite{DBLP:journals/ftdb/IlyasC15,DBLP:conf/icde/ChuIPY14}.
Causal
reasoning~\cite{spirtes2000causation,DBLP:journals/jmlr/PelletE08,DBLP:conf/sigmod/SalimiGS18}
and learning sum-of-product
networks~\cite{DBLP:conf/uai/PoonD11,pmlr-v48-friesen16,DBLP:conf/aaai/0001VMNEK18}
repeatedly discover conditional independencies in the data.
Constraints also arise in many other domains, for example in the
\e{frequent itemset problem}~\cite{DBLP:conf/pods/2005,10.1007/3-540-44503-X_14}, or as
\e{measure based constraints}~\cite{DBLP:journals/is/SayrafiGG08} in
applications like Dempster-Shafer theory, possibilistic theory, and
game theory (see discussion in~\cite{DBLP:journals/is/SayrafiGG08}).
In all these applications, quite often the constraints are learned
from the data, and are not required to hold exactly, but it suffices
if they hold only to a certain degree.

The classical \e{implication problem} asks whether a set of
constraints, called the {\em antecedents}, logically imply another
constraint called the {\em consequent}.  In this setting, both
antecedents and consequent are required to hold exactly, hence we
refer to it as an \e{exact implication} ($\EI$).  The database literature
has extensively studied the $\EI$ problem for integrity constraints and
shown that the implication problem is decidable and axiomatizable for
Functional Dependencies (FDs) and Multivalued Dependencies
(MVDs)~\cite{Maier:1983:TRD:1097039,10.1007/978-3-642-39992-3_17,DBLP:journals/tods/ArmstrongD80,DBLP:conf/sigmod/BeeriFH77},
and undecidable for Embedded Multivalued Dependencies
(EMVDs)~\cite{DBLP:journals/iandc/Herrmann06}. The Armstrong axioms, for instance, provide a sound and complete system for inferring FDs when the antecedents are FDs~\cite{DBLP:journals/tods/ArmstrongD80}, and Beeri et al.~\cite{DBLP:conf/sigmod/BeeriFH77} present a set of axioms which are sound and complete for inferring MVDs and FDs when the antecedents are a set of FDs and MVDs.

\added{
Conditional Independencies
(CI) are assertions of the form $X \perp Y \mid Z$, stating
that $X$ is independent of $Y$ conditioned on $Z$. A CI is \e{saturated} if it covers all variables in the joint distribution (i.e., $XYZ=$ all variables); it is called a \e{conditional} if it has the form $X \fd Y$ and $Y$ is a function of $X$.}
The AI community has
extensively studied the $\EI$ problem for Conditional Independencies (CI), and has shown that
the implication problem is decidable and axiomatizable for saturated
CIs~\cite{GeigerPearl1993} (where $XYZ=$ all variables). On the negative side, it has been shown that in the general case, there exists no finite, sound and complete inference system for the implication of CIs~\cite{StudenyCINoCharacterization1990}. In the Frequent Itemset Problem, a
constraint like $X \rightarrow Y\vee Z \vee U$ means that every basket
that contains $X$ also contains at least one of $Y, Z, U$, and the
implication problem here is also decidable and
axiomatizable~\cite{Sayrafi:2005:DC:1065167.1065213}.

\paragraph{The Relaxation Problem.} In this paper we consider a new problem,
called the {\em relaxation problem}: if an exact implication holds,
does an approximate implication hold too?  For example, suppose we
prove that a given set of FDs implies another FD, but the input data
satisfies the antecedent FDs only to some degree: to what degree does
the consequent FD hold in the database?  
The relaxation
problem asks whether we can convert an exact implication into an \e{approximate implication}.  In other words, we ask how the error in the antecedents propagates to the error in its consequent.
When relaxation
holds, then the error to the consequent can be bounded, and any inference system for proving exact implication,
e.g. using a set of axioms or some algorithm, can be used to infer an
approximate implication.

In order to study the relaxation problem we need to measure the degree
of satisfaction of a constraint.  In this paper we use Information
Theory. This is the natural semantics for modeling CIs of
multivariate distributions, because $X \perp Y \mid Z$ iff
$I(X;Y|Z)=0$ where $I$ is the \e{conditional mutual information}. FDs
and MVDs are special cases of
CIs~\cite{DBLP:journals/tse/Lee87,Dalkilic:2000:ID:335168.336059,DBLP:journals/tsmc/WongBW00}
(reviewed in Section~\ref{subsec:ci:ic}), and thus they are naturally
modeled using the information theoretic measure $I(X;Y|Z)$ or $H(Y|X)$; in contrast, EMVDs do not appear
to have a natural interpretation using information theory, and we will
not discuss them further in this work.
An \e{approximate
	implication} ($\APXI$) is an inequality that (numerically) bounds the information-theoretic measure of the
consequent (e.g., conditional entropy $H(\cdot|\cdot)$ if it is an FD, or the conditional mutual information $I(\cdot;\cdot|\cdot)$ if it is an MVD) by a linear combination of the information-theoretic measures of the antecedents.
The link between integrity constraints and information theoretic measures is, in itself, not new. Several papers have argued that information
theory is a suitable tool to express integrity
constraints~\cite{DBLP:journals/tse/Lee87,Dalkilic:2000:ID:335168.336059,DBLP:journals/tsmc/WongBW00,DBLP:journals/is/Malvestuto86,DBLP:journals/is/GiannellaR04}.

An exact implication ($\EI$) becomes an assertion of the form
$(\sigma_1=0 \wedge \sigma_2=0 \wedge \ldots) \Rightarrow (\tau=0)$,
while an approximate implication ($\APXI$) is a linear inequality
$\tau \leq \lambda\cdot \left(\sum \sigma_i\right)$, where $\lambda \geq 0$, and
$\tau, \sigma_1, \sigma_2, \ldots$ are information theoretic
measures.  We say that a class of constraints can be
{\em relaxed} if every $\EI$ where the antecedents are from this class, implies the corresponding $\APXI$;
we also say that this $\EI$ admits a $\lambda$-relaxation, when we want to specify the factor $\lambda$ in the $\APXI$.  By the non-negativity of the Shannon information measures (described in Section~\ref{sec:notations}), we get that an
$\APXI$ always implies the corresponding~$\EI$.

\paragraph{Results.} We make several contributions, summarized in
Table~\ref{tab:AIResults}.  We start by showing in
Section~\ref{sec:PBoundedRelaxations} that MVDs$+$FDs admit an
$n^2/4$-relaxation, where $n$ is the number of variables.  When the
consequent is an FD, we show that implication admits a 1-relaxation.
Thus, whenever an exact implication holds between MVD$+$FDs, a simple
linear inequality also holds between their associated information
theoretic terms.  In fact, we prove a stronger result that holds for
CIs in general, which implies the result for MVDs$+$FDs. In Section~\ref{sec:imeasure}, we further show that
under some additional syntactic restrictions to the antecedents, the bound can be tightened: from a $n^2/4$-relaxation to a 1-relaxation,
even when the consequent is not an FD. We leave open the question of whether
a 1-relaxation exists in general.

So far, we have restricted ourselves to saturated or conditional CIs
(which correspond to MVDs or FDs).  In
Section~\ref{sec:gamman} we remove all restrictions, and
prove a negative result: there exists an $\EI$ that does not relax
(Eq.~\eqref{eq:kr:sigma:tau}, based on an example
in~\cite{DBLP:journals/tit/KacedR13}).  Nevertheless, we show that
every $\EI$ can be relaxed to its corresponding  $\APXI$ plus an error term, which can be made
arbitrarily small, at the cost of increasing the factor $\lambda$.
This result implies that {\em every} $\EI$ is a consequence of its corresponding inequality $\APXI$, plus an error term.
In fact, the $\EI$ in Eq.~\eqref{eq:kr:sigma:tau} follows
from an inequality by Mat{\'u}{\v s}~\cite{Matus2007}, which is
precisely the associated $\APXI$ plus an error term; our result shows that
every $\EI$ can be proven in this style.

In Section~\ref{sec:BoundedRelaxation} we consider approximate (and exact) implications that can be proved using \e{Shannon's inequalities} (monotonicity and submodularity, reviewed in
Section~\ref{subsec:information:theory}).
In general, Shannon's inequalities are
sound but incomplete for proving exact and approximate
implications that hold for all probability distributions~\cite{DBLP:journals/tit/ZhangY97,DBLP:journals/tit/ZhangY98}, but they are complete for deriving inequalities that hold for all \e{polymatroids} (defined in Section~\ref{subsec:information:theory})~\cite{Yeung:2008:ITN:1457455}. 
We also prove that every exact implication that holds for all polymatroids relaxes to an approximate
implication with a finite upper bound $\lambda \leq (2^n)!$, and a
lower bound $\lambda \geq 3$; the tightness of these bounds remain open.

In Section~\ref{sec:marketbasket} we show that the Shannon inequalities are sound and complete for implication from  \e{measure-based constraints}~\cite{DBLP:journals/is/SayrafiGG08} that arise in market-basket analysis.
More generally,
in Section~\ref{sec:pn} we restrict the class of {\em  models} used to check an
implication, to probability distributions with exactly 2 outcomes (tuples), each with probability 1/2; we justify
this shortly.  We prove that, under this restriction, the implication
problem has a 1-relaxation.  Restricting the models leads to a
complete but unsound method for checking general implication, however
this method is sound for saturated$+$conditional CIs (as we show in
Section~\ref{sec:PBoundedRelaxations}), and is also sound for deriving implications from Frequent Itemset constraints (as we show in Section~\ref{sec:marketbasket}).

In Section~\ref{sec:imeasure} we extend some of the results in this paper and provide alternative proofs using the {\em I-measure} theory that relates between Shannon's information measures and set-theory. Specifically, we extend the result of Section~\ref{sec:PBoundedRelaxations}, and provide an alternative proof to the result of Section~\ref{sec:marketbasket}. 

\paragraph{Two Consequences.} While our paper is focused on relaxation, our
results have two consequences for the exact implication problem. The
first is a {\em 2-tuple model property}: an exact implication, where
the antecedents are saturated or conditional CIs, holds \added{iff it holds on all
uniform probability distributions with 2 tuples.}
A similar result is
known for MVD$+$FDs~\cite{Sagiv:1981:ERD:322261.322263}. Geiger and
Pearl~\cite{GeigerPearl1993}, building on an earlier result by
Fagin~\cite{DBLP:journals/jacm/Fagin82}, prove that every set of CIs
has an \e{Armstrong model}: a discrete probability distribution that
satisfies only the CIs and their consequences, and no other CI.  The
Armstrong model is also called a \e{global witness}, and, in general,
can be arbitrarily large.  Our result concerns
a {\em local witness}:
for any $\EI$, if it fails on some probability distribution, then it
fails on a 2-tuple uniform distribution.

The second consequence concerns the equivalence between the implication
problem of saturated$+$conditional CIs with that of MVD$+$FDs.  It
is easy to check that the former implies the latter
(Section~\ref{sec:notations}).  Wong et
al.~\cite{DBLP:journals/tsmc/WongBW00} prove the other direction (i.e., if an implication from MVDs+FDs holds then the same implication holds for saturated$+$conditional CIs), relying on the sound and complete
axiomatization of MVDs~\cite{DBLP:conf/sigmod/BeeriFH77}. Our
2-tuple model property implies the other direction 
almost immediately, leading to a much simpler proof in Section~\ref{sec:PBoundedRelaxations}.

\added{This article extends the conference publication by the authors~\cite{DBLP:conf/icdt/KenigS20}. We have added in this article all the proofs and intermediate results that were excluded from the conference paper, as well as examples illustrating the ideas and methods introduced.}


\section{Notation and Preliminaries}

\label{sec:notations}

We denote by $[n] = \set{1,2,\ldots,n}$.  If
$\Omega=\set{X_1,\ldots, X_n}$ denotes a set of variables and
$U, V \subseteq \Omega$, then we abbreviate the union $U \cup V$ by
$UV$. 

\subsection{Integrity  Constraints and Conditional Independence}

\label{subsec:ci:ic}

A \e{relation instance} $R$ over signature
$\Omega=\set{X_1,\dots,X_n}$ is a finite set of tuples 
with attributes
$\Omega$.  Let $X, Y, Z \subseteq \Omega$.  We say that the instance
$R$ satisfies the \e{functional dependency} (FD) $X \fd Y$, and write
$R \models X \fd Y$, if for all $t_1, t_2 \in R$, $t_1[X]=t_2[X]$
implies $t_1[Y]=t_2[Y]$.  We say that $R$ satisfies the \e{embedded
  multivalued dependency} (EMVD) $X \mvd Y \mid Z$, and write
$R \models X \mvd Y \mid Z$, if for all $t_1, t_2 \in R$,
$t_1[X]=t_2[X]$ implies $\exists t_3 \in R$ such that
$t_1[XY]=t_3[XY]$ and $t_2[XZ]=t_3[XZ]$.  One can check that
$X \mvd Y \mid Y$ iff $X \fd Y$.  When $XYZ = \Omega$, then we call
$X \mvd Y \mid Z$ a \e{multivalued dependency}, MVD; notice that
$X,Y,Z$ are not necessarily
disjoint~\cite{DBLP:conf/sigmod/BeeriFH77}.

A set of constraints $\Sigma$ {\em implies} a constraint $\tau$, in
notation $\Sigma \Rightarrow \tau$, if for every instance $R$, if
$R \models \Sigma$ then $R \models \impliedCI$.  The implication
problem has been extensively studied in the literature; Beeri et
al.~\cite{DBLP:conf/sigmod/BeeriFH77} gave a complete axiomatization
of FDs and MVDs along with a polynomial time procedure for deciding implication. Herrman~\cite{DBLP:journals/iandc/Herrmann06}
showed that the implication problem for EMVDs is undecidable. In the following we refer to this type of implication as \e{Exact Implication}, abbreviated $\EI$.

Recall that two discrete random variables $X, Y$ are called {\em
  independent} if $p(X=x, Y=y) = p(X=x)\cdot p(Y=y)$ for all outcomes
$x,y$. We say that $X, Y$ are {\em 	conditionally independent} given another discrete random variable $Z$ if $p(X=x, Y=y|Z=z) = p(X=x|Z=z)\cdot p(Y=y|Z=z)$ for all outcomes $x,y$, and $z$.
Fix $\Omega=\set{X_1,\dots,X_n}$
a set of $n$ jointly
distributed discrete random variables with finite domains
$\D_1,\dots,\D_n$, respectively; let $p: \D_1\times \dots \times \D_n\rightarrow [0,1]$ be the probability mass.  
\begin{notation}
\label{notation:Xalpha}
For
$\alpha \subseteq [n]$, denote by $X_\alpha$ the joint random variable
$(X_i: i \in \alpha)$ with domain
$\D_\alpha \defeq \prod_{i \in \alpha} D_i$.	
\end{notation}
We write
$p \models X_\beta \perp X_\gamma | X_\alpha$ when $X_\beta, X_\gamma$
are conditionally independent
given $X_\alpha$. Applying our previous definition, this means that for any triple $(x_\alpha,x_\beta,x_\gamma)$ where $x_\alpha \in \D_\alpha$, $x_\beta \in \D_\beta$ and $x_\gamma \in \D_\gamma$, we have that $p(x_\alpha,x_\beta,x_\gamma)=p(x_\beta|x_\alpha)p(x_\gamma|x_\alpha)$.
In the special case where
$\beta=\gamma$, then $p \models X_\beta\perp X_\beta | X_\alpha$ has the following meaning. Let $x_\alpha\in \D_\alpha$ such that $p(x_\alpha)>0$, and let $x_\beta\in \D_\beta$. Applying the definition of conditional independence directly gives us that $p(x_\beta,x_\beta |x_\alpha)=p(x_\beta|x_\alpha)p(x_\beta|x_\alpha)$.
On the other hand, we know that $p(x_\beta,x_\beta |x_\alpha) \equiv p(x_\beta, |x_\alpha)$. Hence, we get that $p(x_\beta, |x_\alpha)=p(x_\beta|x_\alpha)p(x_\beta|x_\alpha)$. 
This holds only in one of the following conditions: (1) $p(x_\beta|x_\alpha)=0$ or (2) $p(x_\beta|x_\alpha)=1$. In other words, given that $X_\alpha=x_\alpha$ there is exactly one possible assignment $X_\beta=x_\beta$ such that $p(x_\beta|x_\alpha)>0$.  Hence, $X_\alpha$ functionally determines $X_\beta$, and we denote this by $p \models X_\alpha \fd X_\beta$.

An assertion $Y \perp Z | X$ is called a {\em Conditional
  Independence} statement, or a CI; this includes $X \rightarrow Y$ as
a special case.  When $XYZ=\Omega$ we call it \e{saturated}. When $XY \subseteq \Omega$ and $Z=\emptyset$, we call it \e{marginal}. A set of
CIs $\Sigma$ {\em implies} a CI $\tau$, in notation
$\Sigma \Rightarrow \tau$, if every probability distribution that
satisfies $\Sigma$ also satisfies $\tau$.  This implication problem
has also been extensively studied: Pearl and
Paz~\cite{DBLP:conf/ecai/PearlP86} gave a sound but incomplete set of
{\em graphoid axioms}, Studeny~\cite{StudenyCINoCharacterization1990}
proved that no finite axiomatization exists, while Geiger and
Pearl~\cite{GeigerPearl1993} gave a complete axiomatization for
saturated, and marginal CIs.

Lee~\cite{DBLP:journals/tse/Lee87} observed the following connection
between database constraints and CIs. 
The \e{empirical
	distribution} of a relation $R$ is the uniform distribution over its tuples,  in other words, $\forall t\in R$, $p(t) = 1/|R|$.  Then:
\begin{lemC}[\cite{DBLP:journals/tse/Lee87}]\label{lem:MVDMI}
  For all $X,Y,Z \subseteq \Omega$ such that $XYZ = \Omega$.
\begin{align}
  R \models & X \fd Y & \Leftrightarrow &&& p \models X \fd Y & 
\mbox{ and } &&
  R \models & X \mvd Y|Z & \Leftrightarrow &&& p \models (Y \perp Z|X)  \label{eq:mvd:h}
\end{align}
\end{lemC}
\added{As can be seen in the Example of Table~\ref{table:EMVDExample}}, the lemma no longer holds for EMVDs, and
for that reason we no longer consider EMVDs in this paper.  The lemma
immediately implies that if $\Sigma$ is a set of saturated+conditional
CIs and the implication $\Sigma \Rightarrow \tau$ holds for all
probability distributions, then the corresponding implication holds in
databases, where the CIs are interpreted as MVDs or FDs
respectively.  Wong~\cite{DBLP:journals/tsmc/WongBW00} gave a
non-trivial proof for the other direction; we will give a much shorter
proof in Corollary~\ref{corr:MVDToCI}.

\begin{table}[ht]
	\centering
	\added{
	\begin{tabular}{ ccc}
		\toprule
		$X_1$ & $X_2$ & $X_3$  \\ \midrule
		$a$ & $c$ & $e$  \\
		$b$ & $c$ & $e$  \\
		$a$ & $d$ & $e$  \\
		$b$ & $d$ & $e$  \\
		$b$ & $d$ & $f$  \\ \bottomrule
	\end{tabular}
	\vspace{0.2cm}
	\caption{
		\added{
		The relation $R[X_1,X_2,X_3]$ satisfies the EMVD
		$\emptyset \mvd X_1\mid X_2$, 
                yet for the empirical
		distribution,  $I_h(X_1;X_2) \neq 0$
		because $X_1, X_2$ are dependent:
		$p(X_1=a)=2/5 \neq p(X_1=a\mid X_2=c)=1/2$.}}
	\label{table:EMVDExample}
}
\end{table}

\subsection{Background on Information Theory}

\label{subsec:information:theory}

We adopt required notation from the literature on information
theory~\cite{Yeung:2008:ITN:1457455,e13020379}.  For $n > 0$, we
identify vectors in $\real^{2^n}$ with functions
$\pow{[n]}\rightarrow \real$.

\paragraph{Polymatroids.} A function\footnote{Most authors
  consider rather the space $\real^{2^n-1}$, by dropping
  $h(\emptyset)$ because it is always $0$.} $h \in (\realPos)^{2^n}$
  is
called a \emph{polymatroid} if $h(\emptyset)=0$ and it satisfies the
following inequalities, called {\em Shannon inequalities}:
\begin{align}
	\text{Monotonicity: }& h(A)\leq h(B) \text{ for }A\subseteq B \subseteq [n] \label{eq:monotonicity}\\
	\text{Submodularity: }& h(A\cup B)+h(A\cap B) \leq h(A)+h(B) \text{ for all }A,B\subseteq [n] \label{eq:submodularity}
\end{align}
The set of polymatroids is denoted $\Gamma_n \subseteq \left(\realPos\right)^{2^n}$,
and forms a polyhedral cone (reviewed in
Section~\ref{sec:gamman}).  For any polymatroid $h$ and
subsets $A,B,C \subseteq [n]$, we define\footnote{Recall that $AB$
  denotes $A \cup B$.}
\begin{align}
  h(B|A) \eqdef &~h(AB) - h(A) \label{eq:h:cond} \\
  I_h(B;C|A) \eqdef &~h(AB) + h(AC) - h(ABC) - h(A) \label{eq:h:mutual:information}
\end{align}
Then, for all $h\in \Gamma_n$, we have $I_h(B;C|A) \geq 0$ and
$h(B|A) \geq 0$.  The \e{chain rule} for entropies and mutual information, respectively, are the identities:
\begin{align} 
h(BC|D)&=h(B|D)+h(C|BD) 	\label{eq:ChainRuleEnt}\\
	\label{eq:ChainRuleMI}
I_h(B;CD|A)&=I_h(B;C|A)+I_h(B;D|AC)
\end{align}
when the rule being applied is clear from the context, we will only say that we apply the chain rule.

We call $I_h(B;C|A)$ \e{saturated} if $ABC = [n]$, and \e{elemental}
if $|B|=|C|=1$. The conditional entropy $h(B|A)$ is a special case of the mutual information $I_h$,
because
$h(B|A) = I_h(B;B|A)$.

\paragraph{Entropic Functions.} If $X$ is a random variable with
a finite domain $\D$ and probability mass $p$, then $H(X)$ denotes its
entropy
\begin{equation}\label{eq:entropy}
  H(X)\eqdef\sum_{x\in \D}p(x)\log_2\frac{1}{p(x)}
\end{equation}
For a set of jointly distributed random variables
$\Omega=\set{X_1,\dots,X_n}$ we define the function
$h : \pow{[n]} \rightarrow \realPos$ as $h(\alpha) \eqdef H(X_\alpha)$ (see Notation~\ref{notation:Xalpha});
$h$ is called an \e{entropic function}, or, with some abuse, an
\e{entropy}.  The set of entropic functions 
is denoted
$\entropicFunctions_n$.
The quantities $h(B|A)$ and $I_h(B;C|A)$ are called the \e{conditional
  entropy} and \e{conditional mutual information} respectively.  The
conditional independence $p \models B \perp C \mid A$ holds iff
$I_h(B;C|A)=0$, and similarly $p \models A \fd B$ iff $h(B|A)=0$,
thus, entropy provides us with an alternative characterization of
CIs.

\begin{notation}
	\label{notation:CIs}
	We summarize the various ways a conditional independence statement (CI) and related notions are represented in the paper, and the relationships between them. In what follows, we let $p$ be an $n$-variable joint probability distribution over variable set $\Omega$, and we let $H_p$ denote the entropy function corresponding to $p$ (see~\eqref{eq:entropy}). Let $X,Y,Z \subseteq \Omega$ where $X$ and $Y$ are non-empty.
	\begin{enumerate}
		\item The notation $X\bot Y|Z$ implies that for any set of values $x,y,z$ in the domains of $X,Y$ and $Z$ respectively, it holds that $p(X=x,Y=y,Z=z)=p(X=x|Z=z)p(Y=y|Z=z)$.
		\item The statement $X\bot Y|Z$ is equivalent to saying that $I_{H_p}(X;Y|Z)=0$ where $I_{H_p}$ is defined in~\eqref{eq:h:mutual:information}. Note, however, that the notation $I_{h}(X;Y|Z)=0$ is more general since we only assume that $h$ is a polymatroid (i.e., but not necessarily the entropy function associated with a probability distribution).
		\item For a polymatroid $h$ (which may or may not be an entropic function), we denote by the triple $\sigma \eqdef (X;Y|Z)$ a CI statement that may hold either exactly (i.e., $I_h(X;Y|Z)=0$), or approximately (i.e., $I_h(X;Y|Z)\leq \varepsilon$ for some threshold $\varepsilon>0$). We then abbreviate: $h(\sigma)\eqdef I_h(X;Y|Z)$.
		\item When $XYZ=\Omega$, and $p$ is a uniform distribution (i.e., for every $r \in \D_\Omega$ either $p(r)=\frac{1}{M}$ for some $M>0$ or $p(r)=0$), then by Lemma~\ref{lem:MVDMI}, an MVD $Z \mvd X|Y$ holds in $p$ iff $I_{H_p}(X;Y|Z)=0$.
	\end{enumerate}	
\end{notation}
Next, we prove two simple, technical lemmas concerning CIs $(X;Y|Z)$ where the intersection between the variable-sets (e.g., $X\cap Y$, $X\cap Z$, or $Y\cap Z$) may be non-empty. These lemmas will be used later on. 
\begin{lem}
	\label{lem:identicallyZeroMI}
	Let $h\in \Gamma_n$ be an $n$-variable polymatroid, and let $X,Y$  and $Z$ denote subsets of of variables. Then for any CI $(X;Y|Z)$ where $Y\subseteq Z$, it holds that $I_h(X;Y|Z)=0$.
\end{lem}
\begin{proof}
	Since $Y\subseteq Z$, we denote $Z=Z'Y$ where $Z'=Z\setminus Y$. Therefore, by~\eqref{eq:h:mutual:information}:
        \[
		I_h(X;Y|Z'Y)=h(XYZ')+h(Z'Y)-h(Z'Y)-h(XYZ')=0 \qedhere
              \]
\end{proof}
\begin{lem}
		\label{lem:writeAsDisjointMI}
 	Let $h\in \Gamma_n$ be an $n$-variable polymatroid, and let $X,Y$ and $Z$ denote subsets of variables. Then $I_h(X;Y|Z)\equiv h(B_{XY}|Z')+I_h(X';Y'|Z')$ where $X'\eqdef X\setminus YZ$, $Y'\eqdef Y\setminus XZ$, $Z'\eqdef Z\setminus XY$ are pairwise disjoint, and $B_{XY}\eqdef X\cap Y\setminus Z$.
\end{lem}
\begin{proof}
	We define $A_{XY}\eqdef X\cap Y$, $A_{XZ}\eqdef X\cap Z$, $A_{YZ}\eqdef Y\cap Z$, and $A_{XYZ}\eqdef X\cap Y\cap Z$. We let $X'=X\setminus A_{XY}A_{XZ}$, $Y'=Y\setminus A_{XY}A_{YZ}$, and $Z'=Z\setminus A_{XZ}A_{YZ}$. Hence, we can write $I_h(X;Y|Z)$ as:
	{\footnotesize
	\begin{align}		
		&I_h(A_{XY}A_{XZ}X';A_{XY}A_{YZ}Y'|A_{XZ}A_{YZ}Z')& \nonumber \\
		&=I(A_{XZ};A_{XY}A_{YZ}Y'|A_{XZ}A_{YZ}Z')+I(A_{XY}X';A_{XY}A_{YZ}Y'|A_{XZ}A_{YZ}Z') & \text{Chain Rule~\ref{eq:ChainRuleMI}}\nonumber \\
		&=I_h(A_{XY}X';A_{XY}A_{YZ}Y'|A_{XZ}A_{YZ}Z')& \text{Lemma~\ref{lem:identicallyZeroMI}} \nonumber\\
		&=I_h(A_{XY}X';A_{YZ}|A_{XZ}A_{YZ}Z')+I_h(A_{XY}X';A_{XY}Y'|A_{XZ}A_{YZ}Z')& \text{Chain Rule~\ref{eq:ChainRuleMI}}\nonumber \\
		&=I_h(A_{XY}X';A_{XY}Y'|A_{XZ}A_{YZ}Z')& \text{Lemma~\ref{lem:identicallyZeroMI}} \nonumber \\
		&=I_h(A_{XY}X';A_{XY}|Z)+I_h(A_{XY}X';Y'|ZA_{XY})& \text{Chain Rule~\ref{eq:ChainRuleMI}}\nonumber \\
		&=I_h(A_{XY};A_{XY}|Z){+}I_h(X';A_{XY}|A_{XY}Z){+}I_h(X';Y'|ZA_{XY}){+}I_h(A_{XY};Y'|X'ZA_{XY})&\text{Chain Rule~\ref{eq:ChainRuleMI}}\nonumber \\
		&=I_h(A_{XY};A_{XY}|Z)+I_h(X';Y'|ZA_{XY})& \text{Lemma~\ref{lem:identicallyZeroMI}} \nonumber \\
		&=h(A_{XY}|Z)+I_h(X';Y'|ZA_{XY}) & \nonumber \\
		&=h(A_{XYZ}B_{XY}|Z)+I_h(X';Y'|ZA_{XY}) & B_{XY}\eqdef A_{XY}\setminus Z \nonumber \\ 
		&=h(A_{XYZ}|Z)+h(B_{XY}|Z)+I_h(X';Y'|ZA_{XY}) & \text{Chain Rule~\ref{eq:ChainRuleMI}}\nonumber \\
		&=h(B_{XY}|Z)+I_h(X';Y'|ZA_{XY}) & h(A_{XYZ}|Z)=0  \nonumber
	\end{align}
}
By definition, $X'$, $Y'$, $Z$ and $B_{XY}$ are pairwise disjoint, thus proving the claim.
\end{proof}

\paragraph{2-Tuple Relations and Step Functions.} 2-tuple
relations play a key role for the implication problem of MVDs$+$FDs:
if an implication fails, 
then there exists a witness consisting of
only two tuples~\cite{Sagiv:1981:ERD:322261.322263}.  We define a
\e{step function} as the entropy of the empirical distribution of a
2-tuple relation; $R=\set{t_1,t_2}$, $t_1\neq t_2$, and $p(t_1)=p(t_2)=1/2$.  We
denote the step function by $h_U$, where $U \subsetneq \Omega$ is the
set of attributes where $t_1,t_2$ agree. One can check:
\begin{align}
  h_U(W) = &
             \begin{cases}
               0 & \mbox{if  $W \subseteq U$} \\
               1 & \mbox{otherwise}
             \end{cases}
\label{eq:def:step:function}
\end{align}
\added{To see why, observe that, by definition, for every set of attributes $W \subseteq U$ we have that $t_1$ and $t_2$ agree on all attributes in $W$ (i.e., $t_1[W]=t_2[W]$). Hence, $h_U(W)=1 \log 1 =0$. On the other hand, if $W \not\subseteq U$, then $t_1[W]\neq t_2[W]$, and computing the entropy gives us that $h_U(W)=\frac{1}{2} \log 2+ \frac{1}{2} \log 2=1$.}
If we set $U=\Omega$ in (\ref{eq:def:step:function}) then $h_\Omega
\equiv 0$. That is, $h_\Omega(W)=0$ for all $W\subseteq \Omega$, and hence, $h_\Omega$ is equivalent to the function that always returns $0$.
Unless otherwise stated, in this paper we do not consider
$h_\Omega$ to be a step function.  Thus, there are $2^n-1$ step
functions and their set is denoted $\stepfn$.
We will use the following fact extensively in this paper:
$I_{h_U}(Y;Z|X) = 1$ if $X \subseteq U$ and $Y, Z \not\subseteq U$,
and $I_{h_U}(Y;Z|X) = 0$ otherwise.

\begin{figure}
  \centering
 \parbox{0.3\linewidth}{
	\centering
  \begin{tabular}[t]{lllll} \cmidrule[\heavyrulewidth]{1-4}
    $X_1$ & $X_2$ & $U_1$ & $U_2$ & $\pr$ \\ \cmidrule{1-4}
    0 & 0 & 0 & 0 & 1/2 \\
    1 & 1 & 0 & 0 & 1/2 \\ \cmidrule[\heavyrulewidth]{1-4}
  \end{tabular}	
\\
(a)
}
\parbox{0.3\linewidth}{
	\centering
  \begin{tabular}[t]{llll} \cmidrule[\heavyrulewidth]{1-3}
    $X$ & $Y$ & $Z$ & $\Pr$ \\ \cmidrule{1-3}
    0 & 0 & 0 & 1/4 \\
    0 & 1 & 1 & 1/4 \\
    1 & 0 & 1 & 1/4 \\
    1 & 1 & 0 & 1/4 \\ \cmidrule[\heavyrulewidth]{1-3}
  \end{tabular}
\\
(b)
}
\parbox{0.3\linewidth}{
	\centering
  \begin{tabular}{lllll} \cmidrule[\heavyrulewidth]{1-4}
  $A$ & $B$ & $C$ & $D$ & $\Pr$ \\ \cmidrule{1-4}
  0 & 0 & 0 & 0 & $1/2 -\varepsilon$ \\
  0 & 1 & 0 & 1 & $1/2 -\varepsilon$ \\
  1 & 0 & 1 & 0 & $\varepsilon$ \\
  1 & 1 & 0 & 0 & $\varepsilon$ \\ \cmidrule[\heavyrulewidth]{1-4}
\end{tabular}
\\
(c)
}
\caption{Two uniform probability distributions (a),(b); a
  distribution from ~\cite{DBLP:journals/tit/KacedR13} (c).}
  \label{fig:examples}
\end{figure}

We now present two simple polymatroids that correspond to the entropic functions of the distributions in Figure~\ref{fig:examples}.
\begin{exa} \label{ex:entropic} Consider the relational instance
  in Fig.~\ref{fig:examples} (a).  Its entropy is the step function
  $h_{U_1U_2}(W)$, which is 0 for $W \subseteq U_1U_2$ and 1
  otherwise.  $R \models X_1 \fd X_2$ because
  $h(X_2|X_1) = h(X_1X_2) - h(X_1) = 1-1 = 0$, and
  $R \not\models U_1 \fd X_1$ because
  $h(X_1|U_1) = h(X_1U_1)-h(U_1)=1-0\neq 0$.
\end{exa}
\begin{exa}
	\label{ex:entropic2}
  The relational instance $R = \setof{(x,y,z)}{x+y+z \mod 2 = 0}$ in Fig.~\ref{fig:examples} (b) is called the {\em parity function}.\eat{,
  defined as  and its empirical
  distribution.}  Its entropy is $h(X)=h(Y)=h(Z)=1$,
  $h(XY)=h(XZ)=h(YZ)=h(XYZ)=2$. We have that $R \models Y \perp Z$ 
  because
  $I_h(Y;Z)=h(Y)+h(Z)-h(YZ)=1+1-2=0$, but $R \not\models Y \perp Z|X$
  because $I_h(Y;Z|X)=1$. To see this, observe that:
  \begin{align*}
  	I_h(Y;Z|X)&=h(XY)+h(XZ)-h(X)-h(XYZ) \\
  	&=2+2-1-2 \\
  	&=1
  \end{align*}  
\end{exa}

\subsection{Discussion} 

This paper studies exact and approximate implications, expressed as
equalities or inequalities of entropic functions $h$.  For example,
the augmentation axiom for MVDs~\cite{DBLP:conf/sigmod/BeeriFH77}
$A \mvd B | CD \Rightarrow AC \mvd B |D$ is expressed as
$I_h(B;CD|A)=0 \Rightarrow I_h(B;D|AC) = 0$, which holds by the chain
rule (\ref{eq:ChainRuleMI}).  Thus, our golden standard is to prove
that (in)equalities hold for all entropic functions
$h \in \Gamma_n^*$.

\added{Fix a  set $A \subseteq \real^N$ in an $N$-dimensional euclidean space with distance metric $d$. The 
{\em topological closure} of $A$ is defined as:
\begin{align*}
  \cl{A} \defeq & \setof{\sx \in \real^N}{\forall \varepsilon > 0, \exists \bs \in A: d(\bs,\sx) < \varepsilon}
\end{align*}
Equivalently, $\cl{A}$ consists of all limits of convergent sequences
in $A$.  The set $A$ is called {\em topologically closed} if
$A = \cl{A}$.  The topological closure enjoys the following basic
property: if $f : \real^N \rightarrow \real$ is any continuous
function, and $f(\sx) \geq 0$ for all $\sx \in A$, then
$f(\sx) \geq 0$ for $\sx \in \cl{A}$.}
Yeung~\cite{Yeung:2008:ITN:1457455} has proven that,
when $n \geq 3$, then $\entropicFunctions_n$ is not topologically
closed, in other words, $\entropicFunctions_n \subsetneq
\cl{\entropicFunctions_n}$.  The elements of the set
$\cl{\entropicFunctions_n}$ are called {\em almost entropic
  functions}.  Equivalently, 
a function $g$ is \e{almost entropic} if it is the limit of a sequence
of entropic functions.
It follows immediately from our discussion that,
if an {\em inequality} holds for all entropic functions
$h \in \entropicFunctions_n$, then, by continuity,
it also holds for
all almost entropic functions $h \in \cl{\entropicFunctions_n}$.
However, this observation does not extend to {\em implications} of
(in)equalities; Kaced and
Romashchenko~\cite{DBLP:journals/tit/KacedR13} gave an example of an
exact implication 
that holds only for entropic functions but fails for
almost entropic functions.  Thus, when discussing an EI, it matters
whether we assume that it holds for $\entropicFunctions_n$ or for
$\cl{\entropicFunctions_n}$.  The only result in this paper where this
distinction matters are the two main theorems in
Section~\ref{sec:gamman}: the negative result
Theorem~\ref{th:no:approx} holds for both $\entropicFunctions_n$ and
for $\cl{\entropicFunctions_n}$, while the positive result
Theorem~\ref{th:mainConeTheorem} holds only for
$\cl{\entropicFunctions_n}$.  The results in
Section~\ref{sec:PBoundedRelaxations} apply to {\em any} set of
polymatroids $K$ that contains all step functions, i.e.
$\stepfn \subseteq K \subseteq \Gamma_n$, thus they apply to both
$\entropicFunctions_n$ and $\cl{\entropicFunctions_n}$, while those in
Section~\ref{sec:BoundedRelaxation} and Section~\ref{sec:marketbasket} are
stated only for $\Gamma_n$ and only for (the conic closure of)
$\stepfn$ respectively.

\section{Definition of the Relaxation Problem}

\label{sec:problem:def}

We now formally define the relaxation problem.  We fix a set of
variables $\Omega = \set{X_1, \ldots, X_n}$, and consider formulas of
the form $\sigma = (Y;Z|X)$, where $X,Y,Z \subseteq \Omega$, which we
call a \e{conditional independence},
CI; when $Y=Z$ then we write it
as $X \fd Y$ and call it a \e{conditional}.  An \e{implication} is a
formula $\Sigma \implies \impliedCI$, where $\Sigma$ is a set of CIs
called \e{antecedents} and $\tau$ is a CI called \e{consequent}.  For
a CI $\sigma=(B;C|A)$, we define $h(\sigma)\eqdef I_h(B;C|A)$ (see~\eqref{eq:h:mutual:information}), for a
set of CIs $\Sigma$, we define
$h(\Sigma)\eqdef\sum_{\sigma \in \Sigma}h(\sigma)$. 
Fix a set $K$
s.t.  $\stepfn \subseteq K \subseteq \Gamma_n$. We recall that $\stepfn$ is the set of step functions, and that $\Gamma_n$ is the set of polymatroids.

\begin{defi} \label{def:ei:ai} The \e{exact implication} ($\EI$)
  $\Sigma \implies \impliedCI$ holds in $K$, denoted
  $K \models_{\EI} (\Sigma \implies \impliedCI)$ if, for all $h \in K$,
  $h(\Sigma)=0$ implies $h(\tau)=0$.  The \e{$\lambda$-approximate
    implication} ($\lambda$-$\APXI$) holds in $K$, in notation
  $K \models \lambda\cdot h(\Sigma) \geq h(\tau)$,
  if
  $\forall h \in K$, $\lambda\cdot h(\Sigma) \geq h(\tau)$.  The
  \e{approximate implication} holds, in notation
  $K \models_{\APXI} (\Sigma \implies \impliedCI)$, if there exist a
  $\lambda \geq 0$ such that the $\lambda$-$\APXI$ holds.
\end{defi}

We will sometimes consider an equivalent definition for $\APXI$, as
$\sum_{\sigma \in \Sigma} \lambda_\sigma h(\sigma) \geq h(\tau)$,
where $\lambda_\sigma \geq 0$ are coefficients, one for each
$\sigma \in \Sigma$; these two definitions are equivalent, by taking
$\lambda = \max_\sigma \lambda_\sigma$.  Notice that both $\EI$ and $\APXI$
are preserved under subsets of $K$ in the sense that
$K_1 \subseteq K_2$ and $K_2 \models_x (\Sigma \implies \impliedCI)$
implies $K_1 \models_x (\Sigma \implies \impliedCI)$, for
$x \in \set{\EI,\APXI}$.

$\APXI$ always implies $\EI$.  Indeed, $h(\tau) \leq \lambda \cdot h(\Sigma)$
and $h(\Sigma)=0$ imply $h(\tau)\leq 0$, which further implies
$h(\tau)=0$, because $h(\tau)\geq 0$ for every CI $\tau$, and every
polymatroid $h$.  In this paper we study the reverse.

\begin{defi}
  Let $\mathcal{I}$ be a syntactically-defined class of implication
  statements $(\Sigma \Rightarrow \tau)$, and let
  $K \subseteq \Gamma_n$.  We say that $\mathcal{I}$ \e{admits a
    relaxation} in $K$ if every implication statement
  $(\Sigma \Rightarrow \impliedCI)$ in $\mathcal{I}$ that holds
  exactly also holds approximately:
  $K\models_{\EI} \Sigma \Rightarrow \impliedCI$ implies
  $K \models_{\APXI} \Sigma \Rightarrow \impliedCI$.  We say that
  $\mathcal{I}$ admits a \e{$\lambda$-relaxation} in $K$, if every EI admits a
  $\lambda$-$\APXI$.
\end{defi}

\added{ If $K_1 \subseteq K_2$ and $\mathcal{I}$ admits a
  $\lambda$-relaxation in $K_1$ then, in general, it does not
  necessarily admit a $\lambda$-relaxation in $K_2$, nor vice versa.
  However, the following more limited fact holds:

\begin{fact} \label{prop:relax:cl}
  If $\mathcal{I}$ admits a $\lambda$-relaxation in $K$, then it also
  admits a $\lambda$-relaxation in $\cl{K}$.
\end{fact}

\begin{proof}
  Let $\Sigma \implies \tau$ be a CI statement in the syntactic class
  $\mathcal{I}$, and suppose
  $\cl{K} \models_{\EI} (\Sigma \implies \impliedCI)$.  Then
  $K \models_{\EI} (\Sigma \implies \impliedCI)$.  Since $\mathcal{I}$
  admits a $\lambda$-relaxation in $K$, we have $K \models_{\APXI}
  (\Sigma \implies \tau)$, or, equivalently, $\forall h \in K: \lambda
  \cdot h(\Sigma) \geq h(\tau)$.  By continuity, $\forall h \in
  \cl{K}: \lambda h(\Sigma) \geq h(\tau)$, proving the claim.
\end{proof}
}

\begin{exa} 
  Let $\Sigma{=} \set{(A;B|\emptyset),(A;C|B)}$ and
  $\impliedCI{=}(A;C|\emptyset)$. 
  By the chain rule~\eqref{eq:ChainRuleMI}, we have that
  \begin{align*}
  	I_h(A;BC)&=I_h(A;B|\emptyset) + I_h(A;C|B)& \text{ and } \\
  	I_h(A;BC)&=I_h(A;C|\emptyset) + I_h(A;B|C)\geq I_h(A;C|\emptyset)&
  \end{align*}
  Hence, $I_h(A;C|\emptyset) \leq I_h(A;B|\emptyset) + I_h(A;C|B)$.
  Therefore, the exact implication $\entropicPlhdrl_n \models_{\EI}\Sigma\implies \tau$ admits a $1$-$\APXI$.
\end{exa}

\section{Relaxation for FDs and MVDs: Always Possible}\label{sec:PBoundedRelaxations}

In this section we consider the implication problem where the
antecedents are either saturated CIs, or conditionals.  This is a
case of special interest in databases, because the constraints
correspond to MVDs, or FDs.  Recall that
a CI $(B;C|A)$ is
\e{saturated} if $ABC = \Omega$ (i.e., the set of all attributes).
Our main result in this section is:
\begin{thm}\label{thm:MainInGamman}
  Assume that each formula in $\Sigma$ is either saturated or a
  conditional (e.g., $Z\fd X$), and let $\tau$ be an arbitrary CI.  Assume
  $\stepfn \models_{\EI} \Sigma \implies \tau$.  Then:
\begin{enumerate}
	\item \label{item:main:1} $\entropicPlhdrl_n \models \frac{n^2}{4}h(\Sigma) \geq h(\tau)$.
	\item \label{item:main:2} If $\tau$ is a conditional,
          then $\entropicPlhdrl_n \models h(\Sigma) \geq h(\tau)$.
\end{enumerate}
\end{thm}

Before we prove the theorem, we list two important consequences.

\begin{cor}\label{corr:PniffGamman}
  Let $\Sigma$ consist of saturated CIs and/or conditionals, and let
  $\tau$ be any CI. Then $\stepfn \models_{\EI} \Sigma \implies \tau$
  implies $\entropicPlhdrl_n \models_{\EI} \Sigma \implies \tau$. 
\end{cor}
\begin{proof}  
  If $\stepfn \models_{\EI} \Sigma \implies \tau$ then
  $\forall h \in \entropicPlhdrl_n$,
  $h(\tau) \leq \frac{n^2}{4}h(\Sigma)$, thus $h(\Sigma)=0$ implies
  $h(\tau)=0$.
\end{proof}

\added{The corollary implies that if $\Sigma, \tau$ are restricted to
  saturated CIs and/or conditionals, then the Exact Implication
  problem is the same for $\stepfn$ as it is for any other set $K$ where  $\stepfn \subseteq K \subseteq \entropicPlhdrl_n$.
  }

The corollary has an immediate application to the inference problem in
graphical models~\cite{GeigerPearl1993}.  There, the problem is to
check if every probability distribution that satisfies all CIs in
$\Sigma$ also satisfies the CI $\tau$; we have seen that this is
equivalent to
$\entropicFunctions_n \models_{\EI} \Sigma \implies \tau$.  The
corollary states that it is enough that this implication holds on all of the
uniform 2-tuple distributions, i.e.,
$\stepfn \models_{\EI} \Sigma \implies \tau$, because this implies the
(even stronger!) statement
$\entropicPlhdrl_n \models_{\EI} \Sigma \implies \tau$.
Decidability (i.e., for exact implication)
was already known: Geiger and Pearl~\cite{GeigerPearl1993} 
proved that the set of graphoid axioms is sound and complete 
for the case when both
$\Sigma$ and $\tau$ are saturated. Specifically, the equivalence between saturated CIs and MVDs~\cite{GeigerPearl1993} enables the application of the polynomial implication algorithm devised for MVDs~\cite{Beeri:1980:MPF:320613.320614} which, in this setting where both
$\Sigma$ and $\tau$ are saturated, has a run-time complexity of $O(|\Sigma|n^2)$.
Gyssens et
al.~\cite{DBLP:journals/ipl/GyssensNG14} improve this result by dropping any
restrictions on $\tau$.

The second consequence is the following:

\begin{cor}\label{corr:MVDToCI}
Let $\Sigma,\tau$ consist of saturated CIs and/or conditionals.  Then
the following two statements are equivalent:
\begin{enumerate}
\item \label{item:mvd:ci:1} The implication $\Sigma \implies \tau$
  holds, where we interpret $\Sigma, \tau$ as MVDs and/or FDs.
\item \label{item:mvd:ci:2}
  $\entropicPlhdrl_n \models_{\EI} \Sigma \implies \tau$.
\end{enumerate}
\end{cor}
\begin{proof}
  We have shown right after Lemma~\ref{lem:MVDMI} that
  (\ref{item:mvd:ci:2}) implies (\ref{item:mvd:ci:1}).  
  For the
  opposite direction, by Th.~\ref{thm:MainInGamman}, we need only check
  $\stepfn \models_{\EI} \Sigma \implies \tau$, which holds because
  on every uniform probability distribution a saturated CI holds iff
  the corresponding MVD holds, and similarly for conditionals and FDs.
  Since the 2-tuple relation satisfies the implication for MVDs$+$FDs, it
  also satisfies the implication for CIs, proving the claim.
\end{proof}

Wong et al.~\cite{DBLP:journals/tsmc/WongBW00} proved that the
implication for MVDs is equivalent to that of the corresponding
saturated CIs (called there BMVD); they did not consider FDs.  For the
proof in the hard direction, they use the sound and complete
axiomatization of MVDs in~\cite{DBLP:conf/sigmod/BeeriFH77}.  In
contrast, our proof is independent of any axiomatic system, and is
also much shorter.  Finally, we notice that the corollary also implies
that, in order to check an implication between MVDs and/or FDs, it
suffices to check it on all 2-tuple databases: indeed, this is
equivalent to checking $\stepfn \models_{\EI} \Sigma \implies \tau$,
because this implies Item (\ref{item:mvd:ci:2}), which in turn implies
item (\ref{item:mvd:ci:1}).  This rather surprising fact was first
proven in~\cite{Sagiv:1981:ERD:322261.322263}.

The proof of Theorem~\ref{thm:MainInGamman} follows from a series of lemmas, and a Theorem that is of independent interest, that we prove next. Before
proceeding, we note that we can assume w.l.o.g. that $\Sigma$ consists
only of saturated CIs.  Indeed, if $\Sigma$ contains a non-saturated
term, then by assumption it is a conditional, $X \fd Y$, and we will
replace it with two saturated terms: $(Y;Z|X)$ and $XZ \rightarrow Y$,
where $Z = \Omega\setminus XY$.  Denoting $\Sigma'$ the new set of
formulas, we have $h(\Sigma)=h(\Sigma')$, because
$h(Y|X) = I_h(Y;Z|X) + h(Y|XZ)$.  Thus, we will assume w.l.o.g. that
all formulas in $\Sigma$ are saturated.

We say that a CI $(X;Y|Z)$ is \e{elemental} if $|X|=|Y|=1$.
\added{
\begin{lem}
	\label{lem:nsquareDiv4}
	Every CI $\tau = (X;Y|Z)$  can be written as a sum of $m+n_X\cdot n_Y$ elemental terms where $m=|X\cap Y\setminus Z|$, $n_X=|X\setminus YZ|$, and $n_Y=|Y\setminus XZ|$. Furthermore, it holds that $m+n_X\cdot n_Y\leq \frac{n^2}{4}$.
\end{lem}
\begin{proof}
	By Lemma~\ref{lem:writeAsDisjointMI} we have that $h(\tau)\equiv h(B_{XY}|Z)+I_h(X';Y'|Z)$ where $B_{XY}\eqdef X\cap Y\setminus Z$, $X'=X\setminus YZ$, and $Y'=Y\setminus XZ$. By the chain rule, $h(B_{XY}|Z)$ and $I_h(X';Y'|Z)$ can be written as the sum of $m=|B_{XY}|$, and  $|X'|\cdot |Y'|=n_Xn_Y$ elemental terms respectively.
	Since $B_{XY}$, $X'$, and $Y'$ are pairwise disjoint, then $m+n_X+n_Y\leq n$. Therefore:
	\begin{align*}
		m+n_X\cdot n_Y &\leq m+\frac{(n-m)^2}{4}&\\
		&< m+\frac{n^2-2m^2+m^2}{4}& m< n\\
		&=\frac{1}{4}(n^2+m(4-m))\\
		&\leq \frac{n^2}{4}+1
	\end{align*}
Where the last transition follows by observing that $\frac{1}{4}m(4-m)\leq 1$ for all integral $m$ (e.g., for $m\in \set{1,3}$ then $\frac{1}{4}m(4-m)=\frac{3}{4}$, and for $m=2$, $\frac{1}{4}m(4-m)=1$). 
Since the number of elemental terms is integral then $m+n_X\cdot n_Y\leq \frac{n^2}{4}$ as required.
\end{proof}
}

Theorem~\ref{thm:MainInGamman} follows from the next result, which is
also of independent interest. We say that $\sigma$ \e{covers} $\tau$
if all variables in $\tau$ are contained in $\sigma$; for example
$\sigma=(abc;d|e)$ covers $\tau=(cd;be)$.  Then:

\begin{thm} \label{th:mvd:elemental} Let $\tau$ be an elemental
  CI, and suppose each formula in $\Sigma$ covers
  $\tau$.  Then $\stepfn \models_{EI} (\Sigma \implies \tau)$
  implies $\Gamma_n \models h(\tau) \leq h(\Sigma)$.
\end{thm}

Notice that Theorem~\ref{th:mvd:elemental} immediately implies Item (\ref{item:main:1})
of Theorem~\ref{thm:MainInGamman}, because by Lemma~\ref{lem:nsquareDiv4} every $\tau = (Y;Z|X)$ can
be written as a sum of at most $n^2/4$ elemental terms.
In what follows, we prove
Theorem~\ref{th:mvd:elemental}, then use it to prove item
(\ref{item:main:2}) of Theorem~\ref{thm:MainInGamman}.

Finally, we consider whether (\ref{item:main:1}) of
Theorem~\ref{thm:MainInGamman} can be strengthened to a 1-relaxation;
we give in Th.~\ref{thm:DisjointSaturated} below a sufficient
condition, whose proof uses the notion of
I-measure (Section~\ref{sec:imeasure}), and leave open the
question whether 1-relaxation holds in general for implications where
the antecedents are saturated CIs and conditionals.

\begin{defi}\label{def:disjoint}
  We say that two CIs $(X;Y|Z)$ and $(A;B|C)$ are {\em disjoint} if
  at least one of the following four conditions holds: (1)
  $X\subseteq C$, (2) $Y\subseteq C$, (3) $A\subseteq Z$, or (4)
  $B\subseteq Z$.
\end{defi}
If $\tau = (X;Y|Z)$ and $\sigma = (A;B|C)$ are disjoint, then for any
step function $h_W$, it cannot be the case that both $h_W(\tau)\neq 0$
and $h_W(\sigma) \neq 0$.  Indeed, if such a $W$ exists, then
$Z, C \subseteq W$ and, assuming (1) $X \subseteq C$ (the other three
cases are similar), we have $ZX \subseteq W$ thus $h_W(\tau)=0$.  
\def\partitionSigmaTechnicalLemma{Suppose
	$\stepfn \models_{\EI} \Sigma \implies \tau$, where $\tau=(X;Y|Z)$.
	Let $\sigma \in \Sigma$ such that  $\tau, \sigma$ are disjoint (Def.~\ref{def:disjoint}). Then:
	$\stepfn\models_{\EI} \left(\Sigma{\setminus} \set{\sigma}\right)
	\implies \tau$. }
\begin{lem}\label{lem:partitionSigmaTechnicalLemma}
	\partitionSigmaTechnicalLemma
\end{lem}
\begin{proof}
	Let $\Sigma'\eqdef \Sigma\setminus \set{\sigma}$.  Assume by
	contradiction that there exists a step function $h_W$ such that
	$h_W(\Sigma')=0$ and $h_W(\tau)=1$.  Since $\sigma, \tau$ are
	disjoint, $h_W(\sigma)=0$.  Then $h_W(\Sigma)=0$, contradicting the
	assumption $\stepfn \models_{EI} \Sigma \implies \tau$.
\end{proof}

\subsection{Proof of Theorem~\ref{th:mvd:elemental}}
The following holds by the chain rule, and
will be used later on.

\def\chainRuleTechnicalLemma{ Let $\sigma = (A;B|C)$ and
  $\tau = (X;Y|Z)$ be CIs such that $X\subseteq A$, $Y\subseteq B$,
  $C \subseteq Z$ and $Z \subseteq ABC$. Then,
  $\entropicPlhdrl_n \models h(\tau) \leq h(\sigma)$.  }

\begin{lem}\label{lem:chainRuleTechnicalLemma}
	\chainRuleTechnicalLemma
\end{lem}
\added{
\begin{proof}
	Since $Z \subseteq ABC$, we denote by $Z_A=A\cap Z$, $Z_B=B\cap Z$, and $Z_C=C \cap Z$. Also, we define $A'=A\setminus (Z_A\cup X)$ and $B'=B\setminus (Z_B\cup Y)$. So, we have that: $I(A;B|C)=I(Z_AA'X;Z_BB'Y|C)$.
	By the chain rule, we have that:
	\begin{align*}
		I(Z_AA'X;Z_BB'Y|C)&=I(Z_A;Z_B|C)+I(A'X;Z_B|CZ_A)+I(Z_A;B'Y|Z_BC)\\
		&+I(X;Y|CZ_AZ_B)+I(X;B'|CZ_AZ_BY)+I(A';B'Y|CZ_AZ_BX)
	\end{align*}
	Noting that $Z=CZ_AZ_B$, we get that $I(X;Y|Z)\leq I(A;B|C)$ as required.
\end{proof}
}

We now prove theorem~\ref{th:mvd:elemental}.  We use lower case for
single variables, thus $\tau = (x;y|Z)$ because it is elemental.  We
may assume w.l.o.g. that neither $x$ nor $y$ are in $Z$: $x, y \not\in Z$
(otherwise $I_h(x;y|Z)=0$ and the
lemma holds trivially). The {\em deficit} of an elemental CI
$\tau = (x;y|Z)$ is the quantity $|\Omega - Z|$. We prove by induction
on the deficit of $\tau$ that
$\stepfn \models_{\EI} \Sigma \implies \tau$ implies
$\entropicPlhdrl_n \models h(\tau) \leq h(\Sigma)$.

Assume $\stepfn \models_{\EI} (\Sigma \implies \tau)$, and consider the
step function $h_Z$ at $Z$.  Since $h_Z(\tau) = 1$, there exists
$\sigma \in \Sigma$, written $\sigma = (A;B|C)$, such that $h_Z(\sigma) = 1$;
this means that $C \subseteq Z$, and $A,B \not\subseteq Z$.  In
particular $x, y \not\in C$, therefore $x, y \in AB$, because $\sigma$
covers $\tau$.  If $x \in A$ and $y \in B$ (or vice versa), then
$\entropicPlhdrl_n \models h(\tau) \leq h(\sigma)$ by
Lemma~\ref{lem:chainRuleTechnicalLemma}, proving the theorem.
Therefore, we assume w.l.o.g. that $x,y \in A$ and none is in $B$.
Furthermore, since $B \not\subseteq Z$, there exists $u \in B-Z$.

\paragraph{Base case: $\tau$ is saturated.}  Then $u \not\in xyZ$,
contradicting the assumption that $\tau$ is saturated; in other words,
in the base case, it is the case that $x \in A$ and $y \in B$.

\paragraph{Step:}
Let $Z_A= Z {\cap} A$, and $Z_B=Z {\cap} B$. Since $C \subseteq Z$,
and $\sigma=(A;B|C)$ covers $\tau$, then $Z=Z_AZ_BC$.  We also write
$A=xyA'Z_A$ (since $x,y\in A$) and $B=uB'Z_B$. So, we have that
$\sigma = (A;B|C)=(xyA'Z_A; uB'Z_B|C)$, and we use the
chain rule to define $\sigma_1,\sigma_2$:
\begin{align*}
h(\sigma) = & I_h(xyA'Z_A;uB'Z_B|C)=I_h(\underbrace{xyA'Z_A;uZ_B|C}_{\defeq \sigma_1})+I_h(\underbrace{xyA'Z_A;B'|uCZ_B}_{\defeq \sigma_2})
\end{align*}
We also partition $\Sigma$ s.t.  $h(\Sigma) = h(\sigma_1) + h(\Sigma_2)$, where $\Sigma_2 \defeq (\Sigma\setminus \set{\sigma}) \cup \set{\sigma_2}$. Since $\sigma_2$ is a saturated CI (i.e., contains all variables in $\sigma$, which is saturated), then $\Sigma_2$ is saturated as well (i.e., contains only saturated CIs).

\noindent Next, define $\tau' \defeq (x;uy|Z)$ and use the chain rule
to define $\tau_1,\tau_2$:
\begin{equation}
  \label{eq:partition:tau}
h(\underbrace{x;y|Z}_{\tau}) \leq   I_h(\underbrace{x;uy|Z}_{\tau'})=I_h(\underbrace{x;u|Z}_{\defeq \tau_1})+I_h(\underbrace{x;y|uZ}_{\defeq \tau_2})  
\end{equation}
By Lemma~\ref{lem:chainRuleTechnicalLemma},
$\entropicPlhdrl_n \models h(\sigma_1) \geq h(\tau_1)$.  We will prove:
$\stepfn \models_{\EI} \Sigma_2 \implies \tau_2$.  This implies the
theorem, because $\Sigma_2$ is saturated, so by the induction hypothesis, $\entropicPlhdrl_n \models h(\Sigma_2) \geq h(\tau_2)$
(since the deficit of $\tau_2$ is one less than that of $\tau$), and
the theorem follows from
$h(\Sigma)=h(\sigma_1)+h(\Sigma_2) \geq
h(\tau_1)+h(\tau_2)=h(\tau')\geq h(\tau)$.  It remains to prove
$\stepfn \models_{\EI} \Sigma_2 \implies \tau_2$, and we start with a
weaker claim:

\begin{clm}
  $\stepfn \models_{\EI} \Sigma \implies \tau_2$.
\end{clm}

\begin{proof}
  By Lemma~\ref{lem:chainRuleTechnicalLemma} we have that
  $$h(\sigma) = I_h(xyA'Z_A; uB'Z_B|C) \geq I_h(xy;u|Z) =
  I_h(y;u|Z)+I_h(x;u|yZ).$$ Therefore, $\Gamma_n \models_{\EI}\Sigma \implies
  (x;u|yZ)$, and in particular, $\stepfn \models_{\EI} \Sigma \implies (x;u|yZ)$.
  Since $\stepfn \models_{\EI} \Sigma \implies (x;y|Z)$ (i.e., by assumption), then by the chain rule
  we have that $\stepfn \models_{\EI} \Sigma \implies (x;uy|Z) = \tau'$, and the claim
  follows from (\ref{eq:partition:tau}).
\end{proof}

Finally, we prove $\stepfn \models_{\EI} \Sigma_2 \implies
\tau_2$. Assume otherwise, and let $h_W$ be a step function such that
$h_W(\tau_2) = I_{h_W}(x;y|uZ)=1$, and $h_W(\Sigma_2)=0$. This means
that $uZ\subseteq W$.  Therefore $uZ_B \subseteq W$, implying
$I_{h_W}(xyA'Z_A;uZ_B|C) = h_W(\sigma_1) = 0$ (because $uZ_BC \subseteq uZ$).  Therefore,
$h_W(\Sigma) = h_W(\sigma_1)+h_W(\Sigma_2)=0$, contradicting the fact
that $\stepfn \models_{\EI} \Sigma \implies \tau_2$.
A more direct way to see this is to observe that $\sigma_1$ and $\tau$ are disjoint (Definition ~\ref{def:disjoint}). Hence, by Lemma~\ref{lem:partitionSigmaTechnicalLemma} we have that $\stepfn\models_{\EI}\Sigma\setminus\set{\sigma_1}\implies \tau_2$, and the result follows by noting that $\Sigma\setminus\set{\sigma_1}=\Sigma_2$.

\subsection{Proof of Theorem~\ref{thm:MainInGamman} Item~\ref{item:main:2}}
We now prove Item~\ref{item:main:2} of Theorem~\ref{thm:MainInGamman}. That is, we show that if $\tau$ is a conditional $Z \fd X$ and $\stepfn \models_{\EI} \Sigma \implies \tau$, then $\Gamma_n \models h(\Sigma) \geq h(\tau)$.

For the following lemma, we recall that by the chain rule for entropies (see~\eqref{eq:ChainRuleEnt}), we have that for a conditional $\tau=Z\fd uX$ we have that:
\begin{equation}
	\label{eq:useChainEnt}
h(\tau)=h(Z\fd uX)=h(uX|Z)\underbrace{=}_{\eqref{eq:ChainRuleEnt}}h(u|Z)+h(X|uZ)=h(\underbrace{Z\fd u}_{\tau_1})+h(\underbrace{uZ\fd X}_{\tau_2})
\end{equation}

\def\partitionSigmaLemma{ 
Let $\Sigma$ be a set of saturated CIs s.t.
$\stepfn \models_{\EI} \Sigma \implies \tau$.  Suppose
$\tau = (Z \fd uX)$,
and define $\tau_1 = (Z \fd u)$, $\tau_2 = (uZ \fd X)$ (see~\eqref{eq:useChainEnt}); thus,
$h(\tau)=h(\tau_1)+h(\tau_2)$.  Then, there exists $\Sigma_1$ and
$\Sigma_2$ such that: (1) $h(\Sigma)=h(\Sigma_1)+h(\Sigma_2)$. (2) $\Sigma_1$
covers $\tau_1$ and $\stepfn \models_{\EI} \Sigma_1 \implies
\tau_1$. (3) $\Sigma_2$ is saturated and
$\stepfn \models_{\EI} \Sigma_2 \implies \tau_2$.
}
\begin{lem}\label{lem:partitionSigmaLemma}
	\partitionSigmaLemma
\end{lem}
\begin{proof}
  We partition $\Sigma$ into $\Sigma_1$ and $\Sigma_2$ as follows. For
  every $\sigma=(A;B|C)\in \Sigma$, if $u \in C$ then we place
  $\sigma$ in $\Sigma_2$.  Otherwise, assume w.l.o.g that $u\in A$,
  and we write $A=uA_ZA_XA'$ where $A_Z=A\cap Z$, $A_X=A\cap X$, and
  $A'=A{\setminus}\set{uA_ZA_X}$.  We use the chain rule to define
  $\sigma_1,\sigma_2$:
	\begin{equation}\label{eq:partitionSigma}
	 I_h(A;B|C)=I_h(uA_ZA_XA';B|C)=I_h(\underbrace{uA_Z;B|C}_{\defeq \sigma_1})+I(\underbrace{A_XA';B|uA_ZC}_{\defeq \sigma_2})
	\end{equation}
	We place $\sigma_1$ in $\Sigma_1$, and $\sigma_2$ in
        $\Sigma_2$.  We observe that $\sigma_1$ covers $\tau_1$
        (because $Z=A_ZB_ZC_Z\subseteq A_ZBC$) and $\sigma_2$ is
        saturated.  Furthermore, $h(\Sigma_1)+h(\Sigma_2)=h(\Sigma)$.
        We prove $\stepfn \models_{\EI} \Sigma_1 \implies \tau_1$. By assumption,
        $\stepfn \models_{\EI} \Sigma \implies \tau_1 = (Z\fd u)$.  Let any
        $\sigma_2 = (A;B | C) \in \Sigma_2$; since $u\in C$, by
        Lemma~\ref{lem:partitionSigmaTechnicalLemma} we can remove it,
        obtaining
        $\stepfn \models_{\EI} \Sigma \setminus \set{\sigma_2} \implies \tau_1$;
        repeating this process proves $\stepfn \models_{\EI} \Sigma_1 \implies \tau_1$.
        Finally, we prove $\stepfn \models_{\EI} \Sigma_2 \implies \tau_2$.
        By assumption, $\stepfn \models_{\EI} \Sigma \implies \tau_2 = (uZ \fd X)$.  Let any
        $\sigma_1 = (uA_Z;B|C) \in \Sigma_1$; since $uA_Z \subseteq
        uZ$, by 
        Lemma~\ref{lem:partitionSigmaTechnicalLemma} we can remove it,
        obtaining $\stepfn \models_{\EI} \Sigma\setminus\set{\sigma_1} \implies \tau_2$;
        repeating this process proves $\stepfn \models_{\EI} \Sigma_2 \implies \tau_2$.
\end{proof}

We now complete the proof of Theorem~\ref{thm:MainInGamman} item
\ref{item:main:2}. Let $\tau=(Z\fd X)$, and $\Sigma$ be saturated (recall that by our assumption $\Sigma$ contains only saturated CIs).
We
show, by induction on $|X|$, that if
$\stepfn \models_{\EI} \Sigma \implies \tau$ then
$\entropicPlhdrl_n \models h(\tau) \leq h(\Sigma)$.  If $|X|=1$, then
$X=\set{x}$, $h(x|Z)=I_h(x;x|Z)$ is elemental, and the claim follows
from Th.~\ref{th:mvd:elemental}.  Otherwise, let $u$ be any
variable in $X$, write $\tau = (Z \fd uX')$, and apply
Lemma~\ref{lem:partitionSigmaLemma} to $\tau_1=(Z\fd u)$,
$\tau_2=(Zu\fd X')$, which gives us a partition of $\Sigma$ into
$\Sigma_1, \Sigma_2$ such that $h(\Sigma)=h(\Sigma_1)+h(\Sigma_2)$.  On the one hand,
$\stepfn\models_{\EI} \Sigma_1 \implies \tau_1$, and from
Th.~\ref{th:mvd:elemental} we derive $h(\tau_1)\leq h(\Sigma_1)$
(because $\tau_1$ is elemental, and covered by $\Sigma_1$); on the other hand
$\stepfn\models_{\EI} \Sigma_2 \implies \tau_2$ where $\Sigma_2$ is saturated, which implies, by induction,
$h(\tau_2)\leq h(\Sigma_2)$.
The result follows from
$h(\tau)=h(\tau_1)+h(\tau_2)\leq h(\Sigma_1)+h(\Sigma_2)=h(\Sigma)$,
completing the proof of Theorem~\ref{thm:MainInGamman}.

\subsection{A special case}
\def\DisjointSaturated{ Let $\Sigma$ be a set of saturated, pairwise
	disjoint CI terms (Def.~\ref{def:disjoint}), and $\impliedCI$ be a
	saturated CI.  Then,
	$\stepfn \models_{\EI} (\Sigma \Rightarrow \tau)$ implies
	$\entropicPlhdrl_n \models h(\tau) \leq h(\Sigma)$.  } 
In Theorem~\ref{thm:DisjointSaturated} we show that under the assumption that the set of CIs in $\Sigma$ is pairwise
disjoint (Definition~\ref{def:disjoint}), we also obtain a 1-relaxation.
The rather technical proof that relies on the I-measure is deferred to Section~\ref{sec:imeasure}, where we present the I-measure.

\begin{thm}\label{thm:DisjointSaturated}
	\DisjointSaturated        
\end{thm}

\def\cubicCorollary{
	Let $\Sigma$ be a set of saturated CIs, and let $\tau=I(X;Y|Z)$ a CI.
	If $\entropicPlhdrl_n \models \Sigma \implies \tau$ then 
	$I_h(X;Y|Z)\leq|X|\cdot|Y|h(\Sigma)\leq n^2\cdot h(\Sigma)$ for any polymatroid $h$.
}

\def\MonotonicConeDisjoint{
	Let $\Sigma$ be a set of disjoint Shannon information measures, and $\impliedCI$ be a CI. Then the implication $\entropicPlhdrl_n \models_{EI} \Sigma \implies \impliedCI$ admits unit relaxation for monotonic polymatroids.
}

\section{Relaxation for General CIs}

\label{sec:gamman}

We now extend our discussion from saturated CIs and conditionals
(or, equivalently, FDs and MVDs), to arbitrary
Conditional Independence statements.  We prove two results in this
section.  First, we prove that relaxation fails in general, and,
second, we prove that a weaker form of relaxation holds.

In both results the relaxation problem is considered in
$\cl{\entropicFunctions_n}$.  Recall that our golden standard is to
check whether the relaxation problem holds in $\entropicFunctions_n$.
As we saw, when the constraints are restricted to FDs and MVDs, then
the relaxation problem is the same in $\stepfn$, in
$\entropicFunctions_n$, in $\cl{\entropicFunctions_n}$, and in
$\Gamma_n$.  But for general constraints, these relaxation problems
differ.  Our first result in this section (the impossibility result)
also holds in $\entropicFunctions_n$, by Fact.~\ref{prop:relax:cl},
but we leave open the question whether the second result (weak
relaxation) holds in $\entropicFunctions_n$.

We state formally the two results in this section:

\begin{thm} \label{th:no:approx} There exists $\Sigma, \tau$ with
  four variables, such that
  $\cl{\Gamma_4^*} \models_{\EI} (\Sigma \Rightarrow \impliedCI)$ and
  $\cl{\Gamma_4^*} \not\models_{\APXI} (\Sigma \Rightarrow \impliedCI)$.
\end{thm}

\begin{thm} \label{th:mainConeTheorem} Let $\Sigma, \tau$ be arbitrary
  CIs, and suppose
  $\cl{\Gamma_n^*} \models_{\EI} \Sigma \Rightarrow \tau$.
  Then, for every $\varepsilon > 0$ there exists $\lambda > 0$ such
  that, for all $h \in \cl{\Gamma_n^*}$:
\begin{align}
  h(\tau) \leq & \lambda \cdot h(\Sigma) 
     +  \varepsilon \cdot h(\Omega) \label{ineq:ei:ih}
\end{align}
\end{thm}

We will prove both theorems shortly.  Before we do this, however, we
provide some background and context for these theorems, which requires
us to review the concept of {\em cones}.

\subsection{Cones}\label{sec:cones}

Both theorems~\ref{th:no:approx} and~\ref{th:mainConeTheorem} are best
understood when viewed through the lens of convex analysis, in
particular {\em cones}.  We briefly review cones here, and refer
to~\cite{schrijver-book,DBLP:journals/kybernetika/Studeny93,boyd_vandenberghe_2004}
for more details.  Fix some number $N > 0$.  A set
$K \subseteq \real^N$ is called a \emph{cone}, if for every
$\sx \in K$ and $\theta \geq 0$ we have that $\theta \sx \in K$.  A
set $C \subseteq \real^N$ is called \emph{convex} if, for any two
points $\sx_1,\sx_2 \in C$ and any $\theta \in [0,1]$,
$\theta \sx_1+ (1-\theta)\sx_2 \in C$.  Unless otherwise stated, in
this paper every cone will be assumed to be convex.  The intersection
of a, not necessarily finite, set of convex cones is also a convex
cone.  The {\em conic hull} of a set $C \subseteq \real^N$, in
notation $\conehull{C}$, is the smallest convex cone that contains
$C$, or, equivalently, it is the set of vectors of the form
$\theta_1\sx_1+\dots+\theta_k\sx_k$, where $k\geq 0$,
$\sx_1, \ldots, \sx_k \in C$, and $\theta_i \geq 0, \forall i\in [k]$.

Fix a vector $\su \in \real^N$.  The set
$K = \setof{\sx}{\sx \dotprod \su \leq 0}$ is a convex cone called a
{\em linear half-space}.  A \emph{polyhedral cone} is the intersection
of a finite number of linear half-spaces.  Equivalently, $K$ is
polyhedral if
$K = \setof{\sx}{\sx \dotprod \su_1 \leq 0, \dots, \sx \dotprod \su_r
  \leq 0}$, where $\su_1, \dots, \su_r \in \real^N$ are fixed vectors,
or, also equivalently, $K = \setof{\sx}{\sx^T A \leq 0}$ where
$A \in \real^{N \times r}$ is a matrix.  A cone $K$ is called {\em
  finitely generated} if $K = \conehull{C}$ for some finite set
$C \subseteq \real^N$; equivalently, $K$ is finitely generated
if\footnote{Here, and throughout the paper, the inequality
  $\sz \geq 0$ means component wise inequalities.}
$K = \setof{A \sz}{\sz \in \real^m, \sz \geq 0}$, for some matrix
$A \in \real^{N \times m}$.  Results by Farkas, Minkowski, and Weyl
imply that a cone is finitely generated iff it is
polyhedral~\cite[pp.61]{schrijver-book}. 

As we discussed, an entropic function, or a polymatroid
$h : \pow{[n]} \rightarrow \realPos$ can be seen as a vector
$h \in \real^N$, where $N \defeq 2^n$, in other words
$\entropicFunctions_n, \Gamma_n \subseteq \real^N$.  By definition,
$\Gamma_n$ is a polyhedral cone, hence it is a finitely generated,
convex cone.  Yeung~\cite{Yeung:2008:ITN:1457455} has proven that,
when $n \geq 3$, then $\entropicFunctions_n$ is neither convex, nor a
cone, but its topological closure $\cl{\entropicFunctions_n}$ is
always a convex cone.  When $n \leq 3$,
$\cl{\entropicFunctions_n} = \Gamma_n$ and thus
$\cl{\entropicFunctions_n}$ is finitely generated.  For $n \geq 4$,
$\cl{\entropicFunctions_n}$ is not finitely
generated~\cite{Matus2007}.  A conditional entropy
$h(B|A) = h(AB)-h(A)$ is equal to $\su \dotprod h$ where
$\su \in \real^N$ is the vector having $+1$ on the dimension $AB$,
$-1$ on the dimension $A$, and $0$ everywhere else.  Similarly, the
mutual information $I_h(B;C|A)=h(AB)+h(AC)-h(ABC)-h(A)$ equals
$\sv \cdot h$ where $\sv$ is a vector with two $+1$'s corresponding to dimensions $AB$ and $AC$, two $-1$'s corresponding to dimensions $ABC$ and $A$, and the rest 0.

This discussion justifies phrasing the relaxation problem as follows.
Fix a convex cone $K \subseteq \real^N$, and let
$\sy_0, \sy_1, \ldots, \sy_m$ be $m+1$ vectors in $\real^N$.  Relaxation
asks whether statement~\eqref{eq:statement:1} below implies
statement~\eqref{eq:statement:2}:
\begin{align}
  \forall \sx \in K: &&& \sx \dotprod \sy_1 \leq 0 \wedge \cdots \wedge \sx \dotprod \sy_m \leq 0 \Rightarrow  \sx \dotprod \sy_0 \leq 0\label{eq:statement:1}\\
\exists \theta_1, \ldots, \theta_m, \forall \sx \in K: &&& \sx \cdot \sy_0 \leq
 \theta_1 \sx \dotprod \sy_1 + \cdots + \theta_m \sx \dotprod \sy_m\label{eq:statement:2}
\end{align}
When each $\sy_i$ has the property $\sx \dotprod \sy_i \geq 0$ for all
$\sx \in K$ (as is the case with the vectors defining $h(B|A)$ and
$I_h(B;C|A)$), then the condition $\sx \dotprod \sy_i \leq 0$ is
equivalent to $\sx \dotprod \sy_i = 0$, and
statement~\eqref{eq:statement:1} is an Exact Implication.
Furthermore, we can set all $\theta_i$'s in
statement~\eqref{eq:statement:2} to be equal to
$\lambda \defeq \max_i \theta_i$, because $\sx \dotprod \sy_i \geq 0$
implies
$\sum_i \theta_i \sx \dotprod \sy_i \leq \lambda \sum_i \sx \dotprod
\sy_i$, and statement~\eqref{eq:statement:2} is an Approximate
Implication.  If we view relaxation at this level of generality, then
it is easy to find cases where relaxation holds, and where relaxation
fails:

\begin{thm}
  \label{th:relaxation:in:general:cones} (1) If $K$ is finitely
  generated, then statement~\eqref{eq:statement:1} implies
  statement~\eqref{eq:statement:2}.  (2) There exists a convex cone
  $K \subseteq \real^3$ where statement~\eqref{eq:statement:1} does
  not imply statement~\eqref{eq:statement:2}.
\end{thm}

\begin{proof}
  (1) We give here a quick and simple proof based on Farkas lemma; in
  the next section we give in Theorem~\ref{thm:RelaxationInGamman} a
  slightly more elaborate proof, based on the strong duality property
  (which is a consequence of Farkas' lemma), in order to obtain an
  upper bound on the relaxation coefficient. 

  More precisely, we use here the following version of Farkas
  lemma~\cite[pp.61, Corollary 5.3a]{schrijver-book}.  For any matrix
  $A \in \real^{N \times M}$ and vector $\sy \in \real^N$ the
  following two statements are equivalent:\footnote{The first
    statement is normally given with $\geq 0$ instead of $\leq 0$.  It
    is easy to revert the inequality by replacing $A, \sy$ with
    $-A, -\sy$.}
  \begin{itemize}
  \item $\forall \sx \in \real^N$, $\sx^T A \leq 0$ implies $\sx^T \sy
    \leq 0$.
  \item There exists $\sz \in \real^N$, $\sz \geq 0$ such that
    $A \sz = \sy$.
  \end{itemize}
  If $K$ is finitely generated, then it is polyhedral, hence
  $K = \setof{\sx}{\sx^T \su_1 \leq 0 \wedge \cdots \wedge \sx^T \su_r
    \leq 0}$.  Let $A$ be the $N \times (r+m)$ matrix whose columns
  are the vectors $\su_1, \ldots, \su_r, \sy_1, \ldots, \sy_m$.  Then,
  statement~\eqref{eq:statement:1} can be written equivalently as:
  \begin{align*}
    \forall \sx \in \real^N: && \sx^T A \leq 0 \Rightarrow & \sx^T \sy \leq 0
  \end{align*}
  Farkas lemma implies the existence of a vector $\sz \in \real^{r+m}$
  such that $\sz \geq 0$ and $A \sz = \sy$.  Denoting its components
  as $\sz = (\gamma_1, \ldots, \gamma_r, \theta_1, \ldots, \theta_m)$,
  we have $\sy = (\sum_j \gamma_j \su_j)+(\sum_i \theta_i \sy_i)$ and
  therefore:
  \begin{align*}
    \forall \sx \in K: && \sx \dotprod \sy = & (\sum_j \gamma_j \sx \dotprod \su_j)+(\sum_i \theta_i \sx \dotprod \sy_i) \leq  \sum_i \theta_i \sx \dotprod \sy_i
  \end{align*}
  since $\sx \dotprod \su_j \leq 0$, proving the claim.

  (2) Let $K \subseteq \real^3$ be the following cone:
  \begin{align*}
    K= \setof{(x_1,x_2,x_3)}{x_1 \geq 0, x_3 \geq 0, x_1x_3 \geq x_2^2}
  \end{align*}
  Equivalently, $K$ is the {\em positive semidefinite
    cone}~\cite{boyd_vandenberghe_2004}, i.e. the set of semi-positive
  definite $2 \times 2$ matrices:
  \begin{align*}
    K\defeq & \setof{A \defeq \left(\begin{array}{cc} x_1 & x_2 \\ x_2 & x_3 \end{array}\right)}{\forall \su \in \real^2: \su^T A \su \geq 0}
  \end{align*}
  It is immediate to check that $K$ is a convex cone.
  Then $K$ satisfies the following Exact Implication:
  \begin{align*}
    &&&\forall (x_1,x_2,x_3) \in K: \ \ x_1\leq 0 \Rightarrow x_2 \leq 0
  \end{align*}
  because $x_1 \leq 0$ is equivalent to $x_1=0$, implying
  $x_2^2 \leq 0$ thus $x_2 = 0$.  However, $K$ does not satisfy the
  corresponding Approximate Implication, more precisely the following
  is false:
  \begin{align*}
    &&& \exists \lambda > 0, \forall (x_1,x_2,x_3) \in K: \ \  x_2 \leq \lambda x_1
    && \mbox{(this is false)}
  \end{align*}
  Indeed, for every choice of $\lambda > 0$, choose
  $0 < x_1 < 1/\lambda$, and let $x_2 = 1$, $x_3 = 1/x_1$.  Then
  $(x_1,x_2,x_3) \in K$, yet $x_2 > \lambda x_1$.
\end{proof}

\subsection{Proof of Theorem~\ref{th:no:approx}}

For $n\leq 3$, the set $\cl{\entropicFunctions_n}$ is a polyhedral
cone, and relaxation holds, by
Theorem~\ref{th:relaxation:in:general:cones} (1).  Thus, we need a
counterexample with $n=4$ jointly distributed random variables.  The cone
$\cl{\entropicFunctions_4}$ is a subset of $\real^{16}$, hence the
counterexample needed to prove Theorem~\ref{th:no:approx} will be more
complex than that used to prove
Theorem~\ref{th:relaxation:in:general:cones} (2).  For that purpose,
we adapt an example by Kaced and Romashchenko \cite[Inequality
$(\mathcal{I}5')$ and Claim 5]{DBLP:journals/tit/KacedR13}, built upon
an earlier example by Mat{\'u}{\v s}~\cite{Matus2007}.  

Let $\Sigma$ and $\tau$ be the following:
\begin{align}
\Sigma = & \set{(C;D|A), (C;D|B), (A;B), (B;C|D)} &
\tau = &  (C;D) \label{eq:kr:sigma:tau}
\end{align}
We first prove that, for any $\lambda \geq 0$, there exists
an entropic function $h$ such that:
\begin{align}
I_h(C;D) > & \lambda \cdot (I_h(C;D|A) + I_h(C;D|B) +  I_h(A;B) + I_h(B;C|D))
\label{eq:no:ineq}
\end{align}
Indeed, consider the distribution shown in Fig.~\ref{fig:examples} (c)
(from~\cite{DBLP:journals/tit/KacedR13}), where
$0 < \varepsilon < 1/2$.  We will prove that for this distribution, the following identities hold:
\begin{align}
  I_h(C;D) = & \varepsilon + O(\varepsilon^2) \label{eq:i:1} \\
  I_h(C;D|A)= & I_h(C;D|B)=I_h(A;B)=0 \label{eq:i:2} \\
  I_h(B;C|D) = &  O(\varepsilon^2) \label{eq:i:3} 
\end{align}
These three identities prove Eq.(\ref{eq:no:ineq}), by choosing
$\varepsilon$ small enough.  These equalities were stated
in~\cite[Claim 5]{DBLP:journals/tit/KacedR13}, but no proof was
provided; for completeness, we include here their proofs.  In general,
if $X,Y$ are two joint random variables, then the conditional entropy
$h(Y|X)$ satisfies the following equality:
$h(Y|X) = \sum_{x} p(X=x) h(Y|X=x)$, where $x$ ranges over the
outcomes of the random variable $X$, $h(Y|X=x)$ is the standard
entropy of the random variable $Y$ conditioned on $X=x$, and, by
convention, $p(X=x) h(Y|X=x)=0$ when $p(X=x)=0$.  Similarly,
$I_h(Y;Z|X)=\sum_x p(X=x) I_h(Y;Z|X=x)$.  Furthermore, $Y, Z$ are
independent random variables iff $I_h(Y;Z)=0$.  Therefore, the expressions in
Eq.~\eqref{eq:i:2} become:
\begin{align*}
 I_h(C;D|A) = &(1-2\varepsilon)I_h(C;D|A=0) + 2\varepsilon I_h(C;D|A=1) = 0 + 0=0 \\
 I_h(C;D|B) = & \frac{1}{2} I_h(C;D|B=0) + \frac{1}{2} I_h(C;D|B=1) = 0+0 = 0 \\
 I_h(A;B) = & 0
\end{align*}
In the first line we used the fact that when $A=0$ then $C$ is
constant and thus $C, D$ are independent, and when $A=1$ then $D$ is
constant.  The same argument applies to the second line.  For the last
line it suffices to check that $A,B$ are independent:
$p(A=1)=2\varepsilon$, $p(B=1)=1/2$, and
$p(A=1,B=1)=\varepsilon=p(A=1)\cdot p(B=1)$.  (It is standard that the
other combinations of values for $A$ and $B$ follow, for example
$p(A=1,B=0)=p(A=1)-p(A=1,B=1)=p(A=1)-p(A=1)p(B=1)=p(A=1)(1-p(B=1))=p(A=1)p(B=0)$.)
This proves Eq.~\eqref{eq:i:2}.

For the proof of identities~\eqref{eq:i:1} and ~\eqref{eq:i:3}, we
start by listing explicitly the probability distribution of $C,D$,
then expand $I_h(B;C|D) = p(D=0)I_h(B;C|D=0)+p(B;C|D=1)$ and list
explicitly the probability distribution of $B,C$ conditioned on $D=0$
(since $I_h(B;C|D=1)=0$):

\null\hfill
\begin{tabular}{lll}  \cmidrule[\heavyrulewidth]{1-2}
  $C$ & $D$ & \\ \cmidrule{1-2}
  $0$   & $0$   & $1/2$ \\
  $0$   & $1$   & $1/2-\varepsilon$ \\
  $1$   & $0$   & $\varepsilon$ \\ \cmidrule[\heavyrulewidth]{1-2}
\end{tabular}
\hfill
\begin{tabular}{llll}  \cmidrule[\heavyrulewidth]{1-2}
  $B$ & $C$ & \\ \cmidrule{1-2}
  $0$   & $0$   & $(1/2-\varepsilon)/(1/2+\varepsilon)$ & $\defeq 1-\varepsilon_1-\varepsilon_2$\\
  $0$   & $1$   & $\varepsilon/(1/2+\varepsilon)$ & $\defeq \varepsilon_1$ \\
  $1$   & $0$   & $\varepsilon/(1/2+\varepsilon)$ & $\defeq \varepsilon_2$\\
  \cmidrule[\heavyrulewidth]{1-2}
\end{tabular}
\hfill \null 
\newline where $\varepsilon_1 = \varepsilon_2 = O(\varepsilon)$ (we
use different symbols $\varepsilon_1, \varepsilon_2$, even though they
are equal, to make it easier to follow the calculations below).  These
two distributions are quite different (the first has the probability
mass split almost $1/2$ and $1/2$, the second has the entire mass
concentrated on the first outcome) and we need different techniques to
compute the mutual information $I_h(\cdots)$.  We begin with
$I_h(C;D)$:
\begin{align*}
 I_h(C;D) = & h(D) - h(D|C) = h(D) - p(C=0)h(D|C=0) = h(D)-(1-\varepsilon)h(D|C=0)
\end{align*}
For any number $x \in (0,1)$, denote by $B(x)$ a Bernoulli random
variable $X$, with outcomes $p(X=0)=1-x$ and $p(X=1)=x$.  We denote by
$f(x)$ its entropy, and $f'(x), f''(x)$ its derivatives:
\begin{align*}
  f(x) \defeq & -(x \log x + (1-x) \log (1-x)) = - (x \ln x + (1-x)  \ln(1-x)) \log e\\
  f'(x) = & - (\ln x - \ln(1-x)) \log e\\
  f''(x) = & - (\frac{1}{x} + \frac{1}{1-x}) \log e
\end{align*}
Notice that $f'(1/2)=0$, and $f''(1/2) = - 4\log e \approx - 5.77$.
We use Taylor's expansion to compute $f$ in a vicinity of $1/2$.  For
$0 < \theta < 1/2$, there exists $\tau \in (0,\theta)$ such that:
\begin{align*}
  f(1/2-\theta) = & f(1/2) + \theta f'(1/2) + \frac{\theta^2}{2}f''(1/2-\tau)
= 1-O(\theta^2)
\end{align*}
The distribution of $D$ is simply $B(1/2-\varepsilon)$, and the
distribution of $D|_{C=0}$ is
$B\left(\frac{1/2-\varepsilon}{1-\varepsilon}\right)=B(1/2-\delta)$
where $\delta = O(\varepsilon)$, therefore:
\begin{align*}
  I_h(C;D) = & f(1/2-\varepsilon) - (1-\varepsilon)f(1/2-\delta) \\
 = & 1-O(\varepsilon^2) - (1-\varepsilon)(1-O(\delta^2)) \\
 = & \varepsilon + O(\delta^2)-O(\varepsilon^2)
\end{align*}
which proves~\eqref{eq:i:1}.  

Next, we compute $I_h(B;C|D=0)$, for which we apply the formula
\begin{align*}
I(X;Y)= & \sum_{x,y} p(X=x,Y=y) \log\frac{p(X=x,Y=y)}{p(X=x)p(Y=y)}
\end{align*}
We denote by $p_0(-) \defeq p(- | D=0)$, therefore:
\begin{align*}
  I_h(B;C|D=0)
 = &
(1-\varepsilon_1-\varepsilon_2) \log \frac{1-\varepsilon_1-\varepsilon_2}{(1-\varepsilon_2)(1-\varepsilon_1)}+
\varepsilon_1 \log \frac{\varepsilon_1}{(1-\varepsilon_2)\varepsilon_1}+
\varepsilon_2 \log \frac{\varepsilon_2}{\varepsilon_2 (1-\varepsilon_1)}\\
 = & (1-\varepsilon_1-\varepsilon_2) \log \frac{1-\varepsilon_1-\varepsilon_2}{(1-\varepsilon_2)(1-\varepsilon_1)}+
\varepsilon_1 \log \frac{1}{(1-\varepsilon_2)}+
\varepsilon_2 \log \frac{1}{(1-\varepsilon_1)}\\
 = & - \left((1-\varepsilon_1-\varepsilon_2)\ln (1-\varepsilon_0) +
\varepsilon_1 \ln(1-\varepsilon_2)+
\varepsilon_2 \ln(1-\varepsilon_1)
\right) \log e
\end{align*}
where:
\begin{align*}
\varepsilon_0 = & 1 -
\frac{(1-\varepsilon_2)(1-\varepsilon_1)}{1-\varepsilon_1-\varepsilon_2}=\frac{\varepsilon_1\varepsilon_2}{1-\varepsilon_1-\varepsilon_2} =
O(\varepsilon_1\varepsilon_2)
\end{align*}
Finally, we use the fact that $- \ln (1 - x) =
O(x)$ and obtain:
\begin{align*}
  I_h(B;C|D=0) = & (1-\varepsilon_1-\varepsilon_2)O(\varepsilon_0)+\varepsilon_1 O(\varepsilon_2)+\varepsilon_2O(\varepsilon_1)
= O(\varepsilon_1\varepsilon_2)
\end{align*}
which proves equation~\eqref{eq:i:3}

Next, we prove
$\cl{\Gamma_n^*} \models_{\EI} (\Sigma \Rightarrow \impliedCI)$.  This
follows from an inequality initially proven by Mat{\'u}{\v
  s}~\cite{Matus2007}, then adapted
by~\cite{DBLP:journals/tit/KacedR13}.  For completeness, we review
here that inequality, starting with the statement of Theorem 2
in~\cite{Matus2007}, which asserts that for every entropic vector
$h \in \entropicFunctions_5$ over 5 variables, and every natural
number $k \geq 1$:
\begin{align*}
  k\left(\square_{13,24} + \Delta_{34|5} + \Delta_{45|3}\right) +  \Delta_{35|4} + \frac{k(k-1)}{2}\left(\Delta_{24|3}+\Delta_{34|2}\right)\dotprod h \geq & 0
\end{align*}
where, using our notation:
\begin{align*}
 \Delta_{YZ|X} \dotprod h \defeq & I_h(Y;Z|X) &
 \square_{AB,CD}\dotprod h \defeq & I_h(C;D|A) + I_h(C;D|B) + I_h(A;B) - I_h(C;D)
\end{align*}
Substituting $A=1, B=3, C=4, D=2, E=5$ in Mat{\'u}{\v s}'s inequality
and dividing by $k$, we obtain the following (which is Eq.(ii) in
Theorem 2 of~\cite{DBLP:journals/tit/KacedR13}):
\begin{align*}
 & \left(I_h(C;D|A) + I_h(C;D|B) + I_h(A;B) - I_h(C;D)\right) \\
  + & I_h(B;C|E) + I_h(C;E|B) + \frac{1}{k}I_h(B;E|C) + \frac{k-1}{2}\left( I_h(C;D|B) + I_h(B;C|D)\right)\geq 0
\end{align*}
Finally, we set $E=D$ to obtain the following inequality, for all
$h \in \entropicFunctions_n$ and $k \geq 1$:
\begin{align}
  I_h(C;D) \leq & I_h(C;D|A) + \frac{k\vplus3}{2}I_h(C;D|B) + I_h(A;B) + \frac{k{+}1}{2}I_h(B;C|D)
              +  \frac{1}{k}I_h(B;D|C)\label{eq:MatusUnConditionalImplication}
\end{align}

By continuity, the inequality also holds for
$\cl{\entropicFunctions_n}$ too.  We can prove now that the Exact
Implication $\Sigma \Rightarrow \tau$ holds in
$\cl{\entropicFunctions_n}$.  Assume
$I_h(C;D|A)=I_h(C;D|B)=I_h(A;B)=I_h(B;C|D)=0$.  Then
inequality~\eqref{eq:MatusUnConditionalImplication} implies
$0 \leq I_h(C;D) \leq \frac{1}{k}I_h(B;D|C)$, for any $k \geq 1$.
Letting $k \rightarrow \infty$, implies $I_h(C;D)=0$.

It is interesting to observe that inequality
(\ref{eq:MatusUnConditionalImplication}) is almost a relaxation of the
implication (\ref{eq:kr:sigma:tau}): the only extra term is the last
term, which can be made arbitrarily small by increasing $k$.  Our
second result generalizes this.

\subsection{Proof and Discussion of Theorem~\ref{th:mainConeTheorem}}

The proof of Theorem~\ref{th:mainConeTheorem} follows from a more
general statement about cones:
\def\mainConeTheorem{ Let $K\subseteq \R^N$ be a topologically closed,
  convex cone, and let $\sy_0, \sy_1,\dots,\sy_m$ be $m+1$ vectors in
  $\real^N$. The following are equivalent:
  \begin{align}
                               & \forall \sx \in K: & \sx\cdot \sy_1 \leq 0, \dots, \sx\cdot \sy_m  \leq 0 \Rightarrow  \sx \cdot \sy_0 \leq 0 \label{eq:th:main:1}\\
    \forall \varepsilon > 0, \exists \theta_1, \ldots, \theta_m  \geq 0, &\forall \sx \in K, &\sx \dotprod \sy_0 \leq \theta_1 \sx \dotprod \sy_1 + \cdots + \theta_m \sx \dotprod \sy_m + \varepsilon ||\sx||_\infty \label{eq:th:main:2}
  \end{align}
}
\begin{thm}\label{thm:mainConeTheorem}
\mainConeTheorem
\end{thm}

We first show that Theorem~\ref{thm:mainConeTheorem} implies
Theorem~\ref{th:mainConeTheorem}. For this purpose we take
$K \defeq \cl{\entropicFunctions_n}$, which is a closed, convex
cone~\cite{Yeung:2008:ITN:1457455}. Let
$\Sigma = \set{(B_1;C_1|A_1), \ldots, (B_m;C_m|A_m)}$, and
$\tau = (B_0;C_0|A_0)$.  We define the vectors $\sy_i$ such that
$\sy_i \dotprod h = I_h(B_i;C_i|A_i)$, and notice that the conditions
$\sy_i \dotprod h \leq 0$ and $\sy_i \dotprod h = 0$ are equivalent,
because $I_h(B_i;C_i|A_i) \geq 0$ always holds.  If the Exact
Implication $\Sigma \Rightarrow \tau$ holds, then
condition~\eqref{eq:th:main:1} holds.  This implies
condition~\eqref{eq:th:main:2}, and inequality~\eqref{ineq:ei:ih} in
Theorem~\ref{th:mainConeTheorem} follows by setting
$\lambda \defeq \max_i \theta_i$, and observing that
$||h||_\infty = h(\Omega)$ (i.e., this is because $\max_{i\in N}h_i=h_N= h(\Omega)$).

In the rest of this section we prove
Theorem~\ref{thm:mainConeTheorem}.  While we only need the
implication~\eqref{eq:th:main:1} $\Rightarrow$~\eqref{eq:th:main:2},
it helps to observe that the reverse holds too.  Indeed, assuming
$\sx \dotprod \sy_i \leq 0$ for $i=1,m$, the
condition~\eqref{eq:th:main:2} implies that, for every
$\varepsilon > 0$: $\sx \dotprod \sy_0 \leq \varepsilon ||\sx||$.  By
taking $\varepsilon \rightarrow 0$ we obtain
$\sx \dotprod \sy_0 \leq 0$.

To prove the implication~\eqref{eq:th:main:1}
$\Rightarrow$~\eqref{eq:th:main:2} we need to review some properties
of cones.  For any set $C \subseteq \real^N$, its \e{dual},
$C^* \subseteq \real^N$ is the following set:
\begin{equation} \label{eq:dualCone}
	C^* \defeq \set{\sy \mid \forall \sx \in C, \sx\dotprod \sy \geq 0}
\end{equation}
It is immediate to check that $C^*$ is a topologically closed, convex
cone
(because $C^*$ is the intersection of, possibly infinitely many,
linear half-spaces, each of which is a topologically closed, convex
cone), formally: $C^* = \cl{C^*}$. We warn that the $*$ in $\entropicFunctions_n$ does {\em
  not} represent the dual; the notation $\entropicFunctions_n$ for
entropic functions is by now well established, and we adopt it here
too, despite its clash with the standard notation for the dual cone.

We need the following basic  properties of cones:

\begin{enumerate}[(A)]
\item \label{item:cone:3} If $L$ is a finite set, then $\conehull{L}$
  is topologically closed.
\item \label{item:cone:2} For any set $K$, $\cl{\conehull{K}}=K^{**}$.
\item \label{item:cone:1} If $K_1$ and $K_2$ are closed, convex cones
  then:  $(K_1 \cap K_2)^* = \left(\cl{\conehull{K_1^* \cup K_2^*}}\right)$.
\end{enumerate}

Item~\ref{item:cone:3} follows immediately from the definition:
$\conehull{\set{\sx_1, \ldots, \sx_r}} = \setof{\sum_i \theta_i
  \sx_i}{\theta_1, \ldots, \theta_i \in [0,1]}$.
Item~\ref{item:cone:2} is well-known and non-trivial, see for
example~\cite[Fact 6]{DBLP:journals/kybernetika/Studeny93}.

Item~\ref{item:cone:1} is perhaps less well-known, and we include the
proof here.  We will use repeatedly the anti-monotonicity of $(-)^*$:
$C_1 \subseteq C_2$ implies $C_1^* \supseteq C_2^*$.  We first prove
the inclusion:
\begin{align}
  \left(\cl{\conehull{K_1^* \cup K_2^*}}\right) \subseteq & (K_1 \cap K_2)^*
\label{eq:some:label:here:1}
\end{align}
as follows: for $i\in \set{1,2}$, $K_i \supseteq K_1 \cap K_2$, which
implies $K_i^* \subseteq (K_1 \cap K_2)^*$, thus
$K_1^* \cup K_2^* \subseteq (K_1 \cap K_2)^*$.
Eq.~\eqref{eq:some:label:here:1} follows from the fact that
$(K_1 \cap K_2)^*$ is both a convex cone and topologically closed.
Next, we prove:
\begin{align}
(K_1 \cap K_2)^* \subseteq &   \left(\cl{\conehull{K_1^* \cup K_2^*}}\right)
\label{eq:some:label:here:2}
\end{align}
For $i\in \set{1,2}$, we have that
$K_i^* \subseteq \cl{\conehull{K_1^* \cup K_2^*}}$, and therefore it holds that
$K_i^{**} \supseteq \left(\cl{\conehull{K_1^* \cup K_2^*}}\right)^*$.
Since $K_1$ and $K_2$ are closed convex cones then by
item~\ref{item:cone:2} it holds that $K_1^{**}=K_1$ and
$K_2^{**}=K_2$. Therefore, for $i\in \set{1,2}$ we have that
$K_i \supseteq \left(\cl{\conehull{K_1^* \cup K_2^*}}\right)^*$.  From
the above we get that
$K_1 \cap K_2 \supseteq \left(\cl{\conehull{K_1^* \cup
      K_2^*}}\right)^*$, which implies
$(K_1 \cap K_2)^* \subseteq \left(\cl{\conehull{K_1^* \cup
      K_2^*}}\right)^{**}$. By item item~\ref{item:cone:2} we have
that
$\left(\cl{\conehull{K_1^* \cup
      K_2^*}}\right)^{**}=\cl{\conehull{K_1^* \cup K_2^*}}$, proving
Eq.~\eqref{eq:some:label:here:2}.

We now have the tools needed to prove
Theorem~\ref{thm:mainConeTheorem}.  Assume
condition~\eqref{eq:th:main:1} holds, and let $\varepsilon > 0$ be
some number.  Denote by
$L \defeq \set{-\sy_1,-\sy_2, \ldots, -\sy_m}$.  The LHS of the
implication~\eqref{eq:th:main:1} defines the set
$\setof{\sx}{\sx \in K, \sx\dotprod \sy_1 \leq 0, \ldots, \sx \dotprod
  \sy_m \leq 0} = \setof{\sx}{\sx \in K \cap L^*}$.  Thus,
implication~\eqref{eq:th:main:1} states
$x \in K \cap L^* \Rightarrow \sx \dotprod \sy_0 \leq 0$,, which is
equivalent to ${-}\sy_0 \in (K \cap L^*)^*$.  We rewrite this set as:
\begin{align*}
    (K \cap L^*)^* = & \cl{\conehull{K^* \cup L^{**}}} & \text{Item \ref{item:cone:1}}\\
        = & \cl{\conehull{K^* \cup \cl{\conehull{L}}}}  & \text{Item \ref{item:cone:2}}\\
        = & \cl{\conehull{K^* \cup \conehull{L}}}  & \text{Item \ref{item:cone:3}}\\
        = & \cl{\conehull{K^* \cup L}}  & \text{Definition of $\conehull{-}$} 
\end{align*}
Thus, $-\sy_0 \in \cl{\conehull{K^* \cup L}}$.  By the definition of
the topological closure, for all $\delta > 0$, there exists a vector
in $\conehull{K^* \cup L}$ that is $\delta$-close to $-\sy_0$; in
particular, we choose $\delta = \varepsilon/N$.  Equivalently, there
exists a vector $\se \in \real^N$ such that
$||\se||_\infty < \varepsilon/N$ and:
\begin{align*}
-\sy_0+\se \in & \conehull{K^* \cup L}
\end{align*}
By definition of the cone hull, there exists $\su \in K^*$ and
$\theta_1, \ldots, \theta_m \geq 0$ such that:
\begin{align*}
-\sy_0+\se = & \su - \theta_1 \sy_1 - \cdots - \theta_m \sy_m
\end{align*}
or, equivalently:
\begin{align*}
\theta_1 \sy_1 + \cdots + \theta_m \sy_m -\sy_0+\se = & \su \in K^*
\end{align*}
By the definition of $K^*$, it follows that, for every $\sx \in K$:
\begin{align*}
  \theta_1 \sx \dotprod \sy_1 + \cdots + \theta_m \sx \dotprod \sy_m -\sx \dotprod \sy_0+\sx \dotprod \se\geq & 0
\end{align*}
Statement~\eqref{eq:th:main:2} follows by observing that
$|\sx \dotprod \se| \leq N ||\sx||_\infty \cdot ||\se||_\infty <
\varepsilon ||\sx||_\infty$.
$\qedhere$

\section{Restricted Axioms: Shannon Inequalities} \label{sec:BoundedRelaxation}
\added{The results in the previous section are mostly negative: for
  general constraints, relaxation fails, with the only exception of MVDs
  and FDs, where relaxation holds.  In this section we prove that
  relaxation holds in general for constraints that
  can be inferred using only Shannon inequalities
  (monotonicity~\eqref{eq:monotonicity}, and
  submodularity~\eqref{eq:submodularity}).  Equivalently, this means
  interpreting the constraints over the set of all polymatroids, $\Gamma_n$.
}

\added{While the golden standard is the implication problem in  $\cl{\entropicFunctions_n}$,
a study of the implication problem in $\Gamma_n$}
is important for several reasons.  First, by restricting to
Shannon inequalities we obtain a sound, but in general incomplete
(w.r.t. $\cl{\entropicFunctions_n}$)
method for deciding implications. The incompleteness stems from the fact that the 
set of entropic functions (and their limit points) $\closure(\entropicFunctions_n)$ obey additional inequalities (and hence, implications) beyond those that result from the Shannon inequalities. These are called \e{non-Shannon-Type inequalities}~\cite{DBLP:journals/tit/ZhangY97,DBLP:journals/cominfsys/MakarychevMRV02}, and as their name suggests, they are not implied by the Shannon inequalities.
For example,  Mat{\'u}{\v s}'s
inequality~(\ref{eq:MatusUnConditionalImplication}) is a non-Shannon
inequality. It holds only in $\closure(\entropicFunctions_n)$ but fails in $\Gamma_n$.
Second, while generally incomplete, Shannon's inequalities are
complete for characterizing the inequalities that hold under certain
syntactic restrictions.  In particular, it follows from our results in
Section~\ref{sec:PBoundedRelaxations} that they are complete for FDs
and MVDs, and it follows from results
in~\cite{DBLP:journals/corr/abs-2105-14463} that they are complete for
marginal CIs\footnote{Marginal CIs have the form $(X;Y)$ (i.e., no
  conditioning)}: in both cases, relaxation holds, with a factor
$\lambda = n^2/4$.

We have already seen in Theorem~\ref{th:relaxation:in:general:cones}
that  exact implication of CIs relax over $\Gamma_n$.  We start by
proving an upper bound on the coefficient of the relaxation.

\begin{thm}
	\label{thm:RelaxationInGamman}
        Let $\Sigma, \tau$ be arbitrary CIs.  If
        $\Gamma_n \models_{\EI} \Sigma \Rightarrow \tau$, then
        $\Gamma_n \models h(\tau) \leq (2^n)! \cdot h(\Sigma)$.  In
        other words, CIs admit relaxation over $\Gamma_n$ with
        coefficient $\lambda \leq (2^n)!$.
\end{thm}
\added{
\begin{proof}
  We start by proving a lemma:

\begin{lem} \label{lemma:factorial}
  Consider a system of linear equations, $Ax = b$, where $A$ is an
  $M \times N$ matrix with $\texttt{rank}(A)=M$.  Suppose all entries
  in $A$ and $b$ are $-1$, $0$, or $1$.  Then there exists a solution
  $x$ such that, for all $i=1,N$, $|x_i| \leq M!$
\end{lem}

\begin{proof}
  Let $A_0$ be a sub-matrix of $A$ consisting of $M$ linearly
  independent columns.  W.l.o.g. we can assume that these are the
  first $M$ columns, and write $A$ as $A = [A_0 | A_1]$.  The equation
  $Ax=b$ becomes $A_0x_0 + A_1x_1 = b$, where $x_0$ represent the
  first $M$ coordinates of $x$, and $x_1$ are the remaining $N-M$
  coordinates.  We set $x_1=0$, and solve $A_0x_0=b$ using Cramer's
  rule: for all $i=1,M$, $(x_0)_i = \Delta_i / \Delta$, where
  $\Delta = \det(A_0) \neq 0$ and $\Delta_i$ is the determinant of the
  matrix $A_0$ where the $i$'th column is replaced with $b$.  Each
  determinant is the sum of $M!$ terms, and each term is $-1$, $0$, or
  $1$.  It follows that $|\Delta_i| \leq M!$, hence
  $|(x_0)_i| \leq M!/1 = M!$.
\end{proof}

Next, we briefly review two facts from linear programming, see
e.g.~\cite{schrijver-book}.  Let $A \in \R^{M \times N}$,
$b \in \R^M$, $c \in \R^N$.  Then:
  \begin{itemize}
  \item The strong duality theorem states that the primal LP and the
    dual LP have the same optimal values:
    \begin{align*}
      \max \setof{c^T x}{x \in \R^N, x \geq 0, A x  \leq b}=
& \min \setof{y^T b}{y \in \R^M, y \geq 0, y^T A \geq c^T}
    \end{align*}
  \item If the primal linear program has a finite optimal solution,
    then it has an optimal solution $x$ that is a vertex of the
    polytope $\setof{x\in \R^N}{x \geq 0, A x \leq b}$.  In
    particular, if $r \defeq \texttt{rank}(A)$, then there exists an
    optimal solution $x$ of the primal LP that satisfies $A_0 x = b$,
    where $A_0$ is some $r \times N$ sub-matrix consisting of $r$
    independent rows.  As a consequence, if all entries in $A,b$ are
    $-1$ or $0$ or $+1$, then, by Lemma~\ref{lemma:factorial}, there
    exists an optimal solution $x$ satisfying $|x_i| \leq r!$, for
    every coordinate $i=1,N$.
  \end{itemize}

  We prove now Theorem~\ref{thm:RelaxationInGamman}.  Fix a set of $n$
  variables, $\Omega$, and let $\Gamma_n$ be the set of all
  polymatroids over variables $\Omega$.  $\Gamma_n$ is defined by
  Shannon's inequalities, monotonicity~\eqref{eq:monotonicity}, and
  submodularity~\eqref{eq:submodularity}, and it is known that it
  suffices to take only the {\em elemental} inequalities,
  i.e. inequalities of the following form
  (see~\cite[(14.12)]{Yeung:2008:ITN:1457455}):
  \begin{align*}
    \forall X \in \Omega: &&   h(\Omega-X) - h(\Omega) \leq &0 & \mbox{monotonicity}\\
    \forall X, Y \in \Omega, W \subseteq \Omega - \set{X,Y}: && h(W)-h(WX)-h(WY)+h(WXY)\leq & 0 & \mbox{submodularity}
  \end{align*}
  There are $n$ elemental monotonicity constraints, and
  $\frac{n(n-1)}{2}2^{n-2}=n(n-1)2^{n-3}$ submodularity constraints.
  The total number of Shannon inequalities is $n+n(n-1)2^{n-3}$.
  Equivalently, we can write:
  \begin{align}
    \Gamma_n = & \setof{h}{h \in \R^{2^n}, h\geq 0, A_S h \leq 0} \label{eq:def:gamman:lp:constraint}
  \end{align}
  where $A_S$ is the $\left(n+n(n-1)2^{n-3}\right) \times 2^n$ matrix
  corresponding to all Shannon inequalities.  Similarly, the
  constraints $\Sigma$ on $h$ can be defined by some $m \times 2^n$
  matrix $A_\Sigma$, in other words:
  \begin{align*}
    \setof{h}{h \models \Sigma} = & \setof{h}{h \in \R^{2^n}, A_\Sigma h \leq 0}
  \end{align*}
  $A_\Sigma$ has one row for each constraint in $\Sigma$.  For
  example, if $\Sigma$ contains the CI $Z_1;Z_2|V$, then one row in
  the matrix $A_\Sigma$ corresponds to the assertion
  $I_h(Z_1;Z_2|V) \leq 0$, or, equivalently,
  $-h(V)+h(Z_1V)+h(Z_2V)-h(Z_1Z_2V) \leq 0$.  Let
  $M \defeq \left(n+n(n-1)2^{n-3}\right) + m$ and
  $N \defeq 2^n$, and let $A$ be the following $M \times N$ matrix:
  \begin{align*}
    A \defeq & \left[
               \begin{array}{c}
                 A_S \\ \hline
                 A_\Sigma
               \end{array}
\right]
  \end{align*}
  Finally, define $b \in \R^M$ the 0-vector, $b^T = (0,0,\ldots, 0)$,
  and let $c_\tau \in \R^{N}$ be the vector corresponding to the
  constraint $\tau$.  More precisely, assume $\tau= (X;Y|W)$.  Then
  $c_\tau^T h$ is the linear expression $-h(W)+h(WX)+h(WY)-h(WXY)$.

  With these notations, we claim that the Exact Implication
  $\Gamma_n \models_{\EI} (\Sigma \implies \tau)$ holds iff the
  optimal solution of the linear program below has value zero:
  \begin{align*}
    \max \setof{c_\tau^T h}{h \in \R^{2^n}, h\geq 0, Ah \leq b}
  \end{align*}
  Recall that $b = 0$.  Indeed, the constraints
  $h \in \R^{2^n}, h\geq 0, Ah \leq b$ define all polymatorids
  $h \in \Gamma_n$ that satisfy $\Sigma$, while
  $c_\tau^T h = I_h(X;Y|W)$ is the value of consequent $\tau$.  We
  always have
  $\max \setof{c_\tau^T h}{h \in \R^{2^n}, h\geq 0, Ah \leq b} \geq
  0$, because the zero polymatroid $h\defeq 0$ is a feasible solution
  to the LP.  Moreover, if the Exact Implication holds, then every
  feasible solution is a polymatroid that satisfies $\Sigma$, hence it
  satisfies $I_h(X;Y|W) \leq 0$, in other words $c_\tau^T h \leq 0$,
  proving tha the optimal solution to the LP is $\leq 0$.  We conclude
  that, if the Exact Implication holds, then the optimal is $=0$,
  which proves the claim.

  Next, assume that the exact implication holds, thus the primal LP
  has optimal value 0.  By the strong duality theorem, the dual LP
  also has optimal value 0:
  \begin{align*}
    \max \setof{c_\tau^T h}{h \in \R^{2^n}, h\geq 0, Ah \leq b} = & \min \setof{y^T b}{y \in \R^M, y \geq 0, y^T A \geq c_\tau^T} = 0
  \end{align*}
  Since $b=0$, we have $y^Tb=0$, hence, when the exact implication
  holds, then
  $\min \setof{0}{y \in \R^M, y \geq 0, y^T A \geq c_\tau^T}=0$.
  Equivalently, this asserts that the constraint
  $y \geq 0, y^T A \geq c_\tau$ has a feasible solution $y$: otherwise
  the optimal of the dual solution is $\min \emptyset = \infty$.

  We observe now that every feasible solution $y$ to the dual LP
  represents a relaxation of the implication problem.  Indeed, lets
  write $y$ as $y = \left[\begin{array}{c}y_S \\ \hline y_\Sigma\end{array}\right]$, 
  where $y_S$ consists of the first $\left(n+n(n-1)2^{n-3}\right)$
  coordinates, and $y_\Sigma$ of the last $m$ coordinates of $y$.
  Since $y$ is feasible, we have
  $y^T A = y^T_S A_S + y^T_\Sigma A_\Sigma \geq c_\tau^T$.  Let
  $h \in \Gamma_n$ be any polymatroid; since $h \geq 0$, we have:
  \begin{align*}
    c_\tau^T h \leq & (y^T_S A_S + y^T_\Sigma A_\Sigma)h 
= y^T_S (A_Sh) + y^T_\Sigma (A_\Sigma h) \leq y^T_\Sigma (A_\Sigma h)
  \end{align*}
  In the last inequality we used the fact that $A_Sh \leq 0$, by
  definition of $\Gamma_n$ in Eq.~\eqref{eq:def:gamman:lp:constraint}.
  The $m$ coordinates of the vector $A_\Sigma h \in \R^m$ are the
  values $I_h(Z_1;Z_2|V)$, for all $(Z_1;Z_2|V) \in \Sigma$.  Thus,
  $y^T_\Sigma (A_\Sigma h)$ can be written as a positive linear
  combination of the constraints in $\Sigma$, and the inequality
  $c_\tau^T h \leq y^T_\Sigma (A_\Sigma h)$ is precisely an
  Approximate Implication.  If $r \defeq \texttt{rank}(A)$, then by
  our discussion above, we can find a feasible solution $y$ whose
  coordinates are $\leq r! \leq (2^n)!$, since
  $r \leq \min(M,N) = N = 2^n$.  In other words,
  $y \leq (2^n)! (1 \ 1 \ \cdots 1)^T$, where $(1 \ 1 \ \cdots 1)^T$
  is the all-1 vector. Thus, we have:
  \begin{align*}
    I_h(X;Y|W) =  c_\tau^T h \leq & y^T_\Sigma (A_\Sigma h) \leq (2^n)!(1\ 1\ \cdots 1)^T (A_\Sigma h) \leq (2^n)! \sum_{(Z_1;Z_2|V) \in \Sigma} I_h(Z_1;Z_2|V)
  \end{align*}
  This completes the proof of the theorem.
\end{proof}
}

Theorem~\ref{thm:RelaxationInGamman} gives us a very crude upper bound
on the relaxation factor for $\Gamma_n$. Next, we we show a lower bound the factor $\lambda$.  We prove a lower
bound of 3:

\def\thmNoUAI{ 
	The following inequality holds for all polymatroids
        $h\in \Gamma_n$:
        \begin{align}
          h(Z)\leq I_h(A;B|C)+I_h(A;B|D)+I_h(C;D|E)+I_h(A;E)+3h(Z|A)+2h(Z|B)\label{eq:NoUAI} 
        \end{align}
        but the inequality fails if any of the coefficients $3,2$ are
        replaced by smaller values.  In particular, denoting
        $\tau, \Sigma$ the terms on the two sides of
        Eq.(\ref{eq:NoUAI}), the exact implication
        $\Gamma_n \models_{\EI} \Sigma \implies \tau$ holds\footnote{This means that if $h(\Sigma)=0$ where $\Sigma=\set{(A;B|C),(A;B|D),\dots,h(Z|B)}$, then $h(\tau)=h(Z)=0$. That is, $Z$ is deterministic.}, and does
        not have a 1-relaxation.}
\begin{thmC}[\cite{DBLP:journals/corr/abs-0910-0284}] \label{thm:NoUAI}
	\thmNoUAI
\end{thmC}
\added{
\begin{proof}
	We make use the following inequality that was proved in Lemma 1 in~\cite{DBLP:journals/corr/abs-0910-0284}. 
	\begin{equation}\label{eq:noUAILemmaDFZ}
		h(Z|R)+I_h(R;S|T)\geq I_h(Z;S|T)
	\end{equation}
	We apply~\eqref{eq:noUAILemmaDFZ} three times:
	\begin{enumerate}
		\item $h(Z|A)+I_h(A;B|C)\geq I_h(Z;B|C)$
		\item $h(Z|A)+I_h(A;B|D)\geq I_h(Z;B|D)$
		\item $h(Z|A)+I_h(A;E)\geq I_h(Z;E)$ 	
	\end{enumerate}
	Plugging back into the formula we get that: 
	\begin{align}
		&I_h(A;B|C)+I_h(A;B|D)+I_h(C;D|E)+I_h(A;E)+3h(Z|A)+2h(Z|B) \nonumber \\
		&\geq  I_h(Z;B|C)+I_h(Z;B|D)+I_h(Z;E)+I_h(C;D|E)+2h(Z;B) \label{eq:NoUAIEqA}
	\end{align}
	We now apply this identity twice more:
	\begin{enumerate}
		\item $h(Z|B)+I_h(Z;B|C)\geq I_h(Z;Z|C)=h(Z|C)$ 
		\item $h(Z|B)+I_h(Z;B|D)\geq I_h(Z;Z|D)=h(Z|D)$
	\end{enumerate}
	Plugging back into~\eqref{eq:NoUAIEqA} we get that:
	\begin{align*}
		&I_h(Z;B|C)+I_h(Z;B|D)+I_h(Z;E)+I_h(C;D|E)+2h(Z|B) \\
		& \geq h(Z|C)+h(Z|D)+I_h(Z;E)+I_h(C;D|E) \\
		&=h(ZC){-}h(C){+}h(ZD){-}h(D)+h(Z){+}h(E){-}h(ZE){+}h(CE){+}h(DE){-}h(E){-}h(CDE)\\
		&=I_h(Z;E|C)+h(ZEC)+I_h(Z;E|D)+h(ZED)+h(Z)-h(ZE)-h(CDE)\\
		&=I_h(Z;E|C)+I_h(Z;E|D)+I_h(C;D|ZE)+h(CDZE)-h(CDE)+h(Z)\\
		&=I_h(Z;E|C)+I_h(Z;E|D)+I_h(C;D|ZE)+h(Z|CDE)+h(Z)\\
		&\geq h(Z) \qedhere
	\end{align*}
\end{proof}

	We remark that inequality~\eqref{eq:NoUAI} can be verified using known
        tools for testing whether an inequality holds for all polymatroids
        (e.g., ITIP\footnote{\url{https://user-www.ie.cuhk.edu.hk/~ITIP/}}, and
        XITIP\footnote{\url{http://xitip.epfl.ch/}}). 
	It is still open whether the coefficients $3$ and $2$ (for $h(Z|A)$ and $h(Z|B)$, respectively) are tight. To show that, we would need to present a polymatroid for which~\eqref{eq:NoUAI} holds with equality. Using XITIP, we were able to test that the coefficients $3$ and $2$ could not be reduced even by $0.0001$. While this does not rule out the possibility of having the inequality hold for some coefficient $3-\varepsilon$ for a small-enough $\epsilon$, it does allow us to conclude that the exact implication corresponding to~\eqref{eq:NoUAI}  does not have a 1-relaxation.
}

\section{Restricted Models: Positive I-Measure}
\label{sec:marketbasket}

\label{sec:pn}
While relaxation fails in its most general setting, we have
  seen that relaxation holds if we either restrict the type of
  constraints to FDs and MVDs, or if we restrict the implications to
  those that can be inferred from Shannon's inequalities.  In this
  section we consider a different restriction: we will restrict the
  types of models, or databases, over which the constraints are
  interpreted: more precisely we restrict the set of entropic
  functions to the step functions, $\stepfn$, or, equivalently, to
  their conic hull, which we denote by
  $\positiveConen \defeq \conehull{\stepfn}$.  In Section~\ref{sec:imeasure} we show that  these entropic functions are precisely those with a
  {\em positive I-measure}, a notion introduced by
  Yeung~\cite{DBLP:journals/tit/Yeung91,Yeung:2008:ITN:1457455}.  In
  this section, we prove that all EIs admit a 1-relaxation over
  entropic functions with positive I-measure (i.e. over
  $\positiveConen$):

\begin{thm}\label{thm:PositiveConeUnitRelaxationCI}
  Every implication $\Sigma \Rightarrow \tau$ admits a 1-relaxation
  over $\positiveConen$, where $\Sigma, \tau$ are  arbitrary CIs.  In
  other words, if $\positiveConen \models_{\EI} (\Sigma \implies
  \tau)$ then $\forall h \in \positiveConen$, $h(\tau) \leq
  h(\Sigma)$.
\end{thm}

The golden standard for the semantics of constraints is to
  interpret them over $\entropicFunctions_n$, hence the reader may
  wonder what we gain by restricting them to (the conic hull of) the
  step functions, or, equivalently, what we can learn by checking
  implications only on the uniform 2-tuple distributions.  We have two
  motivations.  First, checking an implication only on the uniform
  2-tuple distributions leads to a complete, but unsound procedure for
  checking implication over $\entropicFunctions_n$.  In other words,
  by testing an implication on all uniform 2-tuple distributions we
  can detect if an implication {\em fails}, thus, the procedure is
  complete.  Of course, the procedure is not sound, because an
  implication may hold on all step functions but fail in general.  For
  a simple example, the inequality $I_h(X;Y|Z) \leq I_h(X;Y)$ holds
  for all step functions, but fails on the ``parity function'' in
  Fig.~\ref{fig:examples} (b).  The second motivation is more
  interesting.  It turns out that restricting the models to uniform
  2-tuple distributions leads to a sound and complete procedure in
  some important special cases.  We saw one such case in
  Section~\ref{sec:PBoundedRelaxations}: when the constraints are
  restricted to FDs and MVDs, then in order to check an implication in
  $\entropicFunctions_n$, it suffices to check that it holds on all
  uniform 2-tuple distributions.  We will present in this section a
  second case: checking \e{differential constraints} in market basket
  analysis~\cite{Sayrafi:2005:DC:1065167.1065213}.

In Section~\ref{sec:imeasure} we use the I-measure theory to characterize the conic hull of step functions $\positiveConen$, and provide an alternative proof to the main results of this section: Theorems~\ref{thm:PositiveConeUnitRelaxationCI} and~\ref{thm:PositiveConeUnitRelaxation}. \eat{ that every exact implication that holds for all functions in $\positiveConen$ has a 1-relaxation.
}

\subsection{I-Measure Constraints}

Recall the following definition in information theory:

\begin{defi} \label{def:conditional:multivariate:mi} Fix a set
  $\Omega$ of $n$ variables, and let $h : 2^\Omega \rightarrow \R$ be
  any function.  Let $W, Y \subseteq \Omega$ be two sets of variables.
  Then {\em conditional multi-variate mutual information} is defined
  as:
  \begin{align}
    I_h(Y | W) \defeq & - \sum_{Z: W \subseteq Z  \subseteq W \cup Y}(-1)^{|Z-W|}h(Z)
\label{eq:cond:mvar:mi}
  \end{align}
\end{defi}

The quantity $I_h(Y|W)$ is usually written as
$I_h(y_1; \cdots ;y_m | W)$, where
$Y = \set{y_1, \ldots, y_m}$.  In the literature, $I_h(Y|W)$ is
defined only for entropic functions $h$, but we remove this
restriction here.  When $m=2$ then this is precisely the mutual
information $I_h(y_1;y_2|W) = -h(W)+h(y_1W)+h(y_2W)-h(y_1y_2W)$, and
when $m=1$ then it becomes the conditional
$I_h(y|W) = -h(W)+h(yW) = h(y|W)$.  Notice that
$I_h(Y|W)= I_h(Y-W | W)$.  When $Y\cup W = \Omega$, then we say that
$I_h(Y|W)$ is {\em saturated}.


\begin{defi} \label{def:i:measure:constraint} Fix a set $\Omega$ of
  $n$ variables.  An {\em I-measure constraint statement} is a formula
  of the form $Y|W$, where $W, Y \subseteq \Omega$.  We call the
  constraint \e{saturated} if $Y \cup W = \Omega$.  We say that a
  function $h : 2^\Omega \rightarrow \R$ \e{satisfies} the constraint
  statement, if $I_h(Y|W) =0$.  An \e{implication} is a formula
  $\Sigma \implies (Y|W)$, where $\Sigma$ is a set of I-measure
  constraints and $Y|W$ is an I-measure constraint.  Fix a set $K$
  s.t.  $\stepfn \subseteq K \subseteq \Gamma_n$.  The \e{exact
    implication} $\Sigma \implies (Y|W)$ holds in $K$, in notation
  $K \models_{\EI} \Sigma \implies (Y|W)$, if $\forall h \in K$,
  $\bigwedge_{X|V \in \Sigma} I_h(X|V)=0$ implies $I_h(Y|W)=0$.  The
  \e{$\lambda$-approximate implication} holds in $K$, if
  $\forall h \in K$,  $I_h(Y|W) \leq \lambda \sum_{X|V \in \Sigma} I_h(X|V)$. 
  We say that the exact implication problem
  admits a \e{$\lambda$-relaxation} in $K$, for some $\lambda > 0$, if
  every exact implication $K \models_{\EI} \Sigma \implies (Y|W)$
  implies the $\lambda$-approximate implication $\forall h \in K$,
  $I_h(Y|W) \leq \lambda \sum_{(X|V) \in \Sigma}I_h(X|V)$.
\end{defi}

We prove our main result in this section:

\def\PositiveConeUnitRelaxation{ Exact implications of I-measure
  constraints admit a 1-relaxation in $\positiveConen$.  More
  precisely, if $\Sigma \models (Y|W)$ is an implication of I-measure
  constraints and
  $\positiveConen \models_{\EI} \Sigma \implies (Y|W)$, then
  $\forall h \in \positiveConen$,
  $I_h(Y|W) \leq \sum_{X|V \in \Sigma} I_h(X|V)$.  }
\begin{thm}\label{thm:PositiveConeUnitRelaxation}
	\PositiveConeUnitRelaxation
\end{thm}

Notice that Theorem~\ref{thm:PositiveConeUnitRelaxationCI} follows
immediately as a special case, since every CI is, in particular, an
I-measure constraint.  

To prove theorem~\ref{thm:PositiveConeUnitRelaxation}, we need to
prove two lemmas.  In both lemmas below we will use a simple property.
For any two sets $A, B$:
\begin{align}
  \sum_{C: A \subseteq C \subseteq B} (-1)^{|C-A|} = 
&
  \begin{cases}
    1 & \mbox{ when $A=B$} \\
    0 & \mbox{ otherwise}
  \end{cases}
\label{eq:sum:minus:one}
\end{align}
To prove the identity of~\eqref{eq:sum:minus:one}, it suffices to show that
$\sum_{C: A \subseteq C \subseteq B} (-1)^{|C-A|}=0$ when
$A \subsetneq B$: fix any element $b \in B - A$, and notice that the
sets $C$ that don't contain $b$ are in 1-1 correspondence with those
that do contain $b$, via the mapping $C \mapsto C \cup \set{b}$.
Since the two sets $C$ and $C \cup \set{b}$ have different parities,
their terms cancel out, $(-1)^{|C-A|}+(-1)^{|C\cup\set{b}-A|}=0$, and
the entire sum is zero.

Recall from Section~\ref{sec:notations} that $h_Z$
denotes the ``step function at $Z$'', where $Z \subsetneq \Omega$ (see
the definition in Eq.~\eqref{eq:def:step:function}).  By definition,
their conic hull $\positiveConen \defeq \conehull{\stepfn}$ consists
of all functions of the form $h = \sum_{Z \subsetneq \Omega} d_Z h_Z$,
where $d_Z \geq 0$, for $Z \subsetneq \Omega$ are arbitrary
coefficients.  These coefficients are precisely the saturated
conditional multi-variate mutual information.  This follows from a
more general lemma:

\begin{lem} \label{lemma:basis} Consider the vector space
  $\mathbf{V} = \setof{h}{h : 2^{\Omega} \rightarrow \R,
    h(\emptyset)=0}$.  Then, the $2^n-1$ step functions
  $\setof{h_Z}{Z \subsetneq \Omega}$ are basis for $\mathbf{V}$.
  Moreover, for every function $h \in \mathbf{V}$, denoting by
  $d_Z \defeq I_h(\Omega-Z | Z)$, for all $Z \subsetneq \Omega$ (where
  $I_h$ is defined in Eq.~\eqref{eq:cond:mvar:mi}), we have:
  $h = \sum_{Z \subsetneq \Omega} d_Z h_Z$.  In other words, the
  projections of $h$ on the basis consisting of the step functions are
  precisely the saturated conditional multi-variate mutual
  informations.  
\end{lem}

\begin{proof}
  We prove the following claim: for every $W \subseteq \Omega$, it holds that $$h(W) = \sum_{Z: Z \subsetneq \Omega} d_Z h_Z(W),$$
  where $d_Z \defeq I_h(\Omega-Z|Z)$.  The claim implies that the step
  functions span the entire vector space
  $\setof{h}{h: 2^\Omega \rightarrow \R, h(\emptyset)=0}$.  This
  immediately implies that they are linearly independent, hence are a
  basis for $\mathbf{V}$.  Indeed, the dimensionality of $\mathbf{V}$
  is $2^n-1$, and this coincides with the number of step functions:
  since they span the entire vector space, they must be linearly
  independent.  It remains to prove the claim. 

  To prove the claim, for any $W \subseteq \Omega$, define
  $f(W) \defeq h(\Omega)-h(W)$.  We want to express $f(W)$ as
  $f(W) = \sum_{Z: W \subseteq Z\subseteq \Omega} g(Z)$, for some function
  $g : 2^\Omega \rightarrow \R$.  Such a function $g$ is the unique
  M\"obius inverse of $f$.  Recall that the M\"obius inversion formula
  states that the two expressions below are equivalent:
  \begin{align}
    \forall W: && f(W) = & \sum_{Z: W \subseteq Z\subseteq \Omega} g(Z) & \forall W: && g(W) = & \sum_{Z: W \subseteq Z\subseteq \Omega} (-1)^{|Z-W|}f(Z) \label{eq:mobius}
  \end{align}
  Since $f(\emptyset) = h(\Omega)$, the first equality implies:
  \begin{align*}
    h(\Omega) = & f(\emptyset) = \sum_{Z: Z \subseteq \Omega} g(Z)
  \end{align*}
  Therefore, the left equation in ~\eqref{eq:mobius} can be rewritten
  as:
  \begin{align}
    h(W) = & h(\Omega)-f(W) = \sum_{Z:Z \subseteq \Omega} g(Z) - \sum_{Z: W \subseteq Z \subseteq \Omega} g(Z) = \sum_{Z: W \not\subseteq Z} g(Z)
\label{eq:here:is:h}
  \end{align}
  On the other hand, we can express $g(W)$ in terms of the conditional
  multi-variate information, by expanding $f(Z) = h(\Omega)-h(Z)$ in
  the right equation in~\eqref{eq:mobius}, then using Eq.~\eqref{eq:sum:minus:one}: 
  \begin{align*}
    g(W) = & \left(\sum_{Z: W \subseteq Z} (-1)^{|Z-W|}\right)h(\Omega)-\left(\sum_{Z: W \subseteq Z} (-1)^{|Z-W|}h(Z)\right) \\
     = & 
             \begin{cases}
               h(\Omega)-h(\Omega) = 0 & \mbox{ if $W = \Omega$} \\
               0 + I_h(\Omega-W|W) & \mbox{ if $W \subsetneq \Omega$}
             \end{cases}
  \end{align*}
  We compute now $h(W)$ from~\eqref{eq:here:is:h}:

  \begin{align*}
    h(W) = & \sum_{Z: W \not\subseteq Z} g(Z) =  \sum_{Z: W \not\subseteq Z} I_h(\Omega-Z|Z)
             =  \sum_{Z \subsetneq \Omega} I_h(\Omega-Z|Z) h_Z(W)
  \end{align*}
  where the last equality holds because $h_Z(W) = 1$ when
  $W \not\subseteq Z$ and $h_Z(W)=0$ otherwise.  This proves the
  claim, as required.
\end{proof}

The lemma implies that, if $h \in \positiveConen$, then
$I_h(\Omega-Z|Z) \geq 0$ for all $Z \subsetneq \Omega$.  Indeed, by
definition of the conic hull $\positiveConen = \conehull{\stepfn}$,
the function $h$ can be written as  $h = \sum_{Z: Z \subsetneq \Omega}
d_Z h_Z$, where $d_Z \geq 0$.  By the lemma, we have $d_Z =
I_h(\Omega-Z|Z)$, proving the claim. 

The second lemma is:

\begin{lem} \label{lemma:i:i} For any two sets $W, Y \subseteq \Omega$
  s.t. $Y \neq \emptyset$, the following identity holds:
  \begin{align*}
    I_h(Y|W) = & \sum_{V: W \subseteq V \subseteq \Omega -Y} I_h(\Omega - V|V)
  \end{align*}
\end{lem}

\begin{proof}
  We expand $I_h$ in the RHS according to its definition:
  \begin{align*}
\text{RHS} = 
& \sum_{V: W \subseteq V \subseteq \Omega -Y} I_h(\Omega - V|V) = 
 - \sum_{V: W \subseteq V \subseteq \Omega -Y} \left( \sum_{T: V \subseteq T \subseteq \Omega}(-1)^{|T-V|}h(T)\right)\\
= & - \sum_{T: W \subseteq T} h(T) 
\left(\sum_{V: (W \subseteq V \subseteq \Omega -Y) \wedge (V \subseteq T)} (-1)^{|T-V|}\right)\\
= & - \sum_{T: W \subseteq T} h(T) 
\left(\sum_{V: (W \subseteq V \subseteq (T \cap (\Omega -Y))} (-1)^{|T-V|}\right)
  \end{align*}
By Eq.~\eqref{eq:sum:minus:one}, the inner sum is $=1$ when $W = T
\cap (\Omega-Y)$ and $=0$ otherwise.  The condition  $W = T
\cap (\Omega-Y)$ is equivalent to $W \subseteq T \subseteq W \cup Y$
and thus we obtain:
\begin{align*}
  \text{RHS} = & - \sum_{T: W \subseteq T \subseteq W \cup Y}  (-1)^{|T-W|} h(T)
\end{align*}
The latter expression is equal to $I_h(Y|W)$ by Def.~\ref{def:conditional:multivariate:mi}.
\end{proof}

We will now prove Theorem~\ref{thm:PositiveConeUnitRelaxation}.
Consider an exact implication
$\positiveConen \models_{\EI} \Sigma \implies (Y|W)$.  By
Lemma~\ref{lemma:i:i},
$I_h(Y|W) = \sum_{V: W \subseteq V \subseteq \Omega -Y}
I_h(\Omega-V|V)$.  We claim that, for every set $V$ s.t.
$W \subseteq V \subseteq \Omega -Y$, there exists a constraint
$X|U \in \Sigma$ such that $U \subseteq V \subseteq \Omega-X$.  In
other words, if we expand $I_h(X|U)$ according to
Lemma~\ref{lemma:i:i}, then one of the terms will be
$I_h(\Omega-V|V)$.  To prove the claim, we consider the step function
at $V$, $h_V$, and use the fact that the exact implication must hold
for $h_V$.  We notice that $I_{h_V}(\Omega-V|V) = 1$ and
$I_{h_V}(\Omega-U|U) = 0$ for $U \neq V$: this follows by considering
the expansion of $h_V$ given by Lemma~\ref{lemma:basis},
$h_V =\sum_{U: U \subsetneq \Omega} d_U h_U$, where
$d_U = I_{h_V}(\Omega - U|U)$ and noting that, since $h_V$ is part of
the basis, we have $I_{h_V}(\Omega-V|V)=1$ and $I_{h_V}(\Omega-U|U)=0$
for $U\neq V$.  In particular, $I_{h_V}(Y|W) = 1$, in other words, the
constraint $Y|W$ does not hold in $h_V$.  Since the exact implication
holds, it must be the case that $\Sigma$ does not hold for $h_V$
either, hence there exists an I-constraint $X|U \in \Sigma$ such that
$I_{h_V}(X|U)>0$.  When expanding it according to Lemma~\ref{lemma:i:i},
$I_{h_V}(X|U) = \sum_{T: U \subseteq T \subseteq \Omega-X}
I_{h_V}(\Omega-T|T)$.  Since $I_{h_V}(\Omega-T|T)=0$ for all
$T \neq V$, we conclude that one of the terms must be
$I_{h_V}(\Omega-V|V)$, proving the claim.

We use claim to prove the theorem.  Let $h$ be any function in
$\positiveConen$.  To prove the inequality
$I_h(Y|W) \leq \sum_{X|U \in \Sigma} I_h(X|U)$, we expand both sides:
\begin{align*}
  I_h(Y|W) = & \sum_{V: W \subseteq V \subseteq \Omega - Y} I_h(\Omega-V|V) \\
  \sum_{(X|U) \in \Sigma} I_h(X|U) = & \sum_{T:  U \subseteq T \subseteq \Omega-X, (X|U) \in \Sigma} I_h(\Omega-T|T)
\end{align*}
We have shown that every term $I_h(\Omega-V|V)$ in the first sum
occurs at least once as a term $I_h(\Omega-T|T)$ in the second sum.
Since all terms are $\geq 0$, it follows that
$I_h(Y|W) \leq \sum_{(X|U) \in \Sigma} I_h(X|U)$, as required.

\subsection{\added{Differential Constraints in Market Basket Analysis}
}
We end this section by describing the tight connection between
  I-measure constraints and differential constraints for Market Basket
  Analysis, introduced by Sayrafi and Van
  Gucht~\cite{Sayrafi:2005:DC:1065167.1065213}.

Consider a set of items $\Omega=\set{X_1,\dots,X_n}$, and a set of
baskets $\B = \set{b_1, \ldots, b_N}$ where every basket is a subset
$b_i \subseteq \Omega$. The \e{support function}
$f^{\B}:2^\Omega \rightarrow \mathbb{N}$ assigns to every subset
$W\subseteq \Omega$ the
number of baskets in $\B$ that contain the set $W$:
\begin{align}
  f^{\B}(W) = & |\setof{i}{i\in [n], W \subseteq b_i}| \label{eq:support:function}
\end{align}
The function $f^{\B}$ is anti-monotone: $W_1 \subseteq W_2$ implies
$f^{\B}(W_1) \geq f(W_2)$.  Similarly define $d^{\B}(W)$ the number of
baskets in $\B$ that are equal to the set $W$:
\begin{align*}
  d^{\B}(W) = & |\setof{i}{i\in [n], W = b_i}|
\end{align*}
Then the following two identities are easily verified:
\begin{align}
  \forall W \subseteq \Omega: && f^{\B}(W) = & \sum_{Z: W \subseteq Z} d^{\B}(Z) &
   d^{\B}(W) \defeq & \sum_{Z: W \subseteq Z} (-1)^{|Z-W|}f^{\B}(Z) \label{eq:mobius:f:d}
\end{align}
The first identity follows immediately from the definitions of
$f^{\B}$ and $d^{\B}$.  The second is M\"obius' inversion function.
Building on the identity~\eqref{eq:mobius:f:d}, Sayrafi and Van Gucht
define the \e{density} of any function $f$ as follows:

\begin{defi}
  Let $f:2^\Omega \rightarrow \R$ be any function.  Its \e{density} is
  the unique function $d_f: 2^\Omega \rightarrow \R$ is defined as:
  \begin{align*}
    \forall W \subseteq \Omega: && d_f(W) \defeq & \sum_{Z: W \subseteq Z} (-1)^{|Z-W|}f(Z)
  \end{align*}
\end{defi}
When $f$ is the support function $f^{\B}$ associated to a set of
baskets $\B$, then its density $d_f$ is precisely the function
$d^{\B}$, because it satisfies the left equation
in~\eqref{eq:mobius:f:d}.

For any function $d : 2^\Omega \rightarrow \R$ and any two sets of
items $W, Y \subseteq \Omega$, we denote by
$d(Y|W) \defeq \sum_{V: W \subseteq V \subseteq W\cup Y} d_f(V)$.

A {\em differential constraint} is an expression $Y|W$.  A set of
baskets $\B$ \e{satisfies} the differential constraint\footnote{The
  notation used by Sayrafi and Van
  Gucht~\cite{Sayrafi:2005:DC:1065167.1065213} is $W \rightarrow Y$.
  They actually define slightly more general differential constraints,
  where $Y$ is allowed to be a set of sets of items instead of a set
  of items.  To simplify the presentation, in this paper we only
  consider $Y$ as a set of items.}, if $d^{\B}(Y|W)=0$.  Sayrafi and
Van Gucht~\cite{Sayrafi:2005:DC:1065167.1065213} define an
implication, to be an assertion $\Sigma \implies (Y|W)$, where
$\Sigma$ is a set of differential constraints.  In addition to that,
we also define here an approximate implication:

\begin{defi} \label{def:differential:constraints} Fix a set of items
  $\Omega$.  An \e{implication} is a formula $\Sigma \implies (Y|W)$
  where $\Sigma$ is a set of differential constraints, and $Y|W$ is
  one differential constraint.  The \e{exact implication}
  $\Sigma \implies (Y|W)$ holds, in notation
  $\models_{\EI} (\Sigma \implies (Y|W))$ if, for every set of baskets
  $\B$, $\bigwedge_{(X|V)\in \Sigma} d^{\B}(X|V)=0$ implies
  $d^{\B}(Y|W)=0$.  The \e{$\lambda$-approximate implication} holds,
  if, for every set of baskets $\B$,
  $d^{\B}(Y|W) \leq \lambda \sum_{(X|V) \in \Sigma} d^{\B}(X|V)$.  We
  write $\models_{\APXI} (\Sigma \implies (Y|W))$ if there exists a finite
  $\lambda \geq 0$ such that the $\lambda$-approximate implication
  holds.
\end{defi}

In order to check any implication, it suffices to check whether it
holds for \e{singleton} sets of baskets, $\B = \set{b}$.  This follows
from the linearity of the functions $f^{\B}$ and $d^{\B}$: if
$\B = \set{b_1, \ldots, b_N}$, then
$f^{\B}(W) = \sum_{i=1,N} f^{\set{b_i}}(W)$ and
$d^{\B}(W) = \sum_i d^{\set{b_i}}(W)$.  An Exact Implication
$\Sigma \implies (Y|W)$ holds, iff for every singleton sets of baskets
$\B = \set{b}$, $\bigwedge_{(X|V) \in \Sigma} d^{\set{b}}(X|V)=0$
implies $d^{\set{b}}(Y|W)=0$, and similarly for a
$\lambda$-approximate implication.

With this observation, we can now explain the connection to
information theory. For any $b \subseteq \Omega$, we
denote as usual $h_b$ the step function at $b$ when $b \neq \Omega$,
and denote $h_b = 0$ (the constant function $0$) when $b = \Omega$.
Then:

\begin{thm}
  Fix a set $\Omega$ of $n$ items, and let $\B = \set{b}$ be a set
  consisting of a single basket, $b \subseteq \Omega$.  The following
  hold:
  \begin{itemize}
  \item For all $W \subseteq \Omega$: $f^{\set{b}}(W) = 1 - h_b(W)$.
  \item For all $W \subseteq \Omega$:
    $d^{\set{b}}(W) = I_{h_b}(\Omega-W|W)$.
  \item Every exact implication $\models_{\EI} \Sigma \implies W$
    relaxes to a 1-approximate implication:
    $d^{\set{b}}(W)\leq \sum_{V \in \Sigma}d^{\set{b}}(V)$.
  \end{itemize}
\end{thm}

\begin{proof}
  The first item follows immediately from the definitions:
  $f^{\set{b}}(W)=1$ iff $W \subseteq b$, iff $h_b(W) = 0$.  The
  second item follows also from the definitions:
  \begin{align*}
    d^{\set{b}}(W) = & \sum_{Z: W \subseteq Z} f^{\set{b}}(W) =
\sum_{Z: W \subseteq Z} (1-h_b(W)) =
                       \begin{cases}
                         1 - h_b(\Omega)=0=I_{h_b}(0|\Omega) & \mbox{ if $W=\Omega$}\\
                         0 + I_{h_b}(\Omega-W|W) & \mbox{ otherwise}
                       \end{cases}
  \end{align*}
  (We used Eq.~\eqref{eq:sum:minus:one}.)  Finally, the last item
  follows immediately from the second item and from
  Theorem~\ref{thm:PositiveConeUnitRelaxation}.
\end{proof}

\begin{exa}
  Consider five items $\Omega = \set{A,B,C,D,E}$.  The differential
  constraint $C|AB$ asserts that ``every customer who bought items
  $A,B$ also bought item $C$''.  Indeed, if a set of baskets $\B$
  satisfies the constraint $C|AB$ then
  $d^{\B}(C|AB) = d^{\B}(AB)+d^{\B}(ABD)+d^{\B}(ABE)+d^{\B}(ABDE)=0$;
  this simply says that no basket exists that contains both $A$ and
  $B$ but does not contain $C$.  For a second example, consider the
  differential constraint $CD|AB$.  It asserts ``every customer who
  bought items $A,B$ also bought either $C$ or $D$, or both'', because
  $d^{\B}(CD|AB)= d^{\B}(AB)+d^{\B}(ABE)=0$ implies that no basket
  exists that contains $AB$ and none of $C$ or $D$.

  Consider now the following implication:
  $(CD | \emptyset) \implies (CD|AB)$.  Obviously, the exact
  implication holds, and the 1-approximate implication holds too,
  because
  $d^{\B}(CD|AB) = d^{\B}(AB)+d^{\B}(ABE) \leq d^{\B}(CD|\emptyset) =
  d^{\B}(\emptyset)+d^{\B}(A)+d^{\B}(B)+d^{\B}(E)+d^{\B}(AB)+d^{\B}(AE)+d^{\B}(BE)+d^{\B}(ABE)$.
\end{exa}

\def\propgpolymatroid{
	Let $f$ be the support function associated with a set of baskets $\B$ over a set of items $U=\set{U_1,\dots,U_n}$. Then the function $g:2^U \rightarrow \mathbb{N} \cup \set{0}: g(X)\eqdef |\B|-f(X)$ is a polymatroid.
}

\section{Additional Results and Proofs that rely on the I-Measure}
\label{sec:imeasure}
\eat{\subsubsection{The I-measure}\label{sec:imeasure}}
The I-measure, developed by Yeung~\cite{DBLP:journals/tit/Yeung91,Yeung:2008:ITN:1457455}, is a theory that establishes a one-to-one correspondence between Shannon's information measures (i.e., entropy, conditional entropy, mutual information, and conditional mutual information) and set theory. In Section~\ref{sec:ProofFromSec4}, we use the I-measure to prove Theorem~\ref{thm:DisjointSaturated}. In Section~\ref{sec:Pn1Relaxations} we provide an alternative proof to Theorems~\ref{thm:PositiveConeUnitRelaxationCI} and~\ref{thm:PositiveConeUnitRelaxation}.
We begin by briefly describing the main ideas and theorems of the I-measure theory 
that we use in the remainder of this section; for further details and examples, see chapter 3 of~\cite{Yeung:2008:ITN:1457455}.

Let $h\in \entropicPlhdrl_n$ denote a polymatroid defined over the variables $\set{X_1,\dots,X_n}$. As mentioned, the I-Measure theory formulates a one-to-one correspondence between Shannon's information measures and set theory. Hence, every variable $X_i$ of the polymatroid, is associated with a set $\iset(X_i)$.
The universal set is $\Lambda \eqdef \bigcup_{i=1}^n\iset(X_i)$, the union of all sets associated with the variables. The complement of the set $\iset(X_i)$ is $\isetc(X_i)\eqdef \Lambda\setminus \iset(X_i)$. 
Let $\I \subseteq [n]$. We recall that $X_\I$ is the joint random variable
$X_\I\eqdef (X_i: i \in \I)$ (see Notation~\ref{notation:Xalpha}), and we denote $\iset(X_\I)\eqdef \bigcup_{i\in \I}\iset(X_i)$.
\begin{defi}\label{def:field}
	The field $\mathcal{F}_n$ generated by sets $\iset(X_1),\dots,\iset(X_n)$ is the collection of sets which can be obtained by any sequence of usual set operations (union, intersection, complement, and difference) on $\iset(X_1),\dots,\iset(X_n)$.
\end{defi}

The \e{atoms} of $\mathcal{F}_n$ are sets of the form $\bigcap_{i=1}^nY_i$, where $Y_i$ is either $\iset(X_i)$ or $\isetc(X_i)$. We denote by $\mathcal{A}$ the atoms of $\mathcal{F}_n$.
Note that by our choice of the universal set, $\Lambda \eqdef \bigcup_{i=1}^n\iset(X_i)$, the atom $\bigcap_{i=1}^n\isetc(X_i)$ degenerates to the empty set. This is because:
\begin{equation}
\bigcap_{i=1}^n\isetc(X_i)=\left(\iset(X_1)\cup \cdots \cup \iset(X_n)\right)^c=\Lambda^c=\emptyset
\end{equation}
\eat{We consider only atoms in which at least one set appears in positive form (i.e., the atom $\bigcap_{i=1}^n\isetc(X_i)\eqdef \emptyset$ is defined to be empty).} Hence, there are $2^n-1$ non-empty atoms in $\mathcal{F}_n$. \eat{and $2^{2^n-1}$
sets in $\mathcal{F}_n$ expressed as the union of its atoms.}

We call a function $\measure:\mathcal{F}_n \rightarrow \real$ \e{set additive} if, for every pair of disjoint sets $A$ and $B$, it holds that $\measure(A\cup B)=\measure(A)\vplus\measure(B)$.
A real function $\measure$ defined on $\mathcal{F}_n$ is called
a \e{signed measure} if it is set additive, and $\measure(\emptyset)=0$.

The $I$-measure $\imeasure$ on $\mathcal{F}_n$ is defined by
\begin{equation}
	\label{eq:IMeasureDef}
\imeasure(\iset(X_\I))=H(X_\I)
\end{equation}
where $\I \subseteq \set{1,\dots,n}$, and $H$ is the entropy (see~\eqref{eq:entropy}).
\begin{table}[]
	\centering
	\begin{tabular}{cc}	
		\toprule
		Information Measures & $\imeasure$ \\
                \midrule
		$H(X)$	& $\imeasure(\iset(X))$ \\
		$H(XY)$	& $\imeasure\left(\iset(X)\cup\iset(Y)\right)$ \\
		$H(X|Y)$	& $\imeasure\left(\iset(X)\cap \isetc(Y)\right)$
                \\
		$I_H(X;Y)$ & $\imeasure\left(\iset(X)\cap \iset(Y)\right)$  \\
		$I_H(X;Y|Z)$	&  $\imeasure\left(\iset(X)\cap \iset(Y) \cap
                \isetc(Z)\right)$ \\ \bottomrule
	\end{tabular}
	\vspace{0.2cm}
	\caption{Information measures and associated I-measure
        }
	\label{tab:ImeasureSummary}
\end{table}
Theorem~\ref{thm:YeungUniqueness}~\cite{DBLP:journals/tit/Yeung91,Yeung:2008:ITN:1457455} establishes the one-to-one correspondence between Shannon's information measures and $\imeasure$.
\begin{thmC}[\cite{DBLP:journals/tit/Yeung91,Yeung:2008:ITN:1457455}] \label{thm:YeungUniqueness}
	$\imeasure$ is the unique signed measure on $\mathcal{F}_n$ which is consistent with all Shannon's information measures (i.e., entropies, conditional entropies, mutual information, and conditional mutual information). 	
\end{thmC}

Let $X,Y,Z \subseteq \Omega$ denote three subsets of variables; $Z$ may be empty. We refer to objects of the form $(X)$ and $(X|Z)$ as entropy and conditional entropy terms, respectively. Likewise, we refer to triples of the form $(X;Y|Z)$ as conditional mutual information terms. Collectively, we refer to these as information-theoretic terms. Let $\sigma=(X;Y|Z)$ be a mutual information term.
We denote by $\iset(\sigma)=\iset(X)\cap\iset(Y)\cap\isetc(Z)$ the (I-Measure) set associated with $\sigma$ (see~\cite{Yeung:2008:ITN:1457455} for intuition and examples). Table~\ref{tab:ImeasureSummary} summarizes the extension of the I-measure $\imeasure$ to the rest of the Shannon measures (i.e., conditional entropy, mutual information, and conditional mutual information). For a set of mutual information and entropy terms
$\Sigma$, we let:
\[
\iset(\Sigma)\eqdef\bigcup_{\sigma \in \Sigma}\iset(\sigma).
\]

\begin{thmC}[\cite{Yeung:2008:ITN:1457455}]\label{thm:YeungBuildMeasure}
Let $\Omega=\set{X_1,\dots,X_n}$, and let $\mathcal{F}_n$ be the field generated by sets $\iset(X_1),\dots, \iset(X_n)$. Let $\mu_+: \mathcal{F}_n \rightarrow \real_+$ be any set-additive function that assigns non-negative values to the atoms of $\mathcal{F}_n$. Then the function $H: 2^\Omega \rightarrow \real_+$, defined as $H(X_\alpha)=\mu_+(\iset(X_\alpha))$ is a polymatroid (i.e., meets the polymatroid inequalities~\eqref{eq:monotonicity} and~\eqref{eq:submodularity}).
\end{thmC}	

\def\implicationInclusion{
	Let $\Sigma$ denote a set of CIs. If ~$\entropicPlhdrl_n \models_{\EI} \Sigma \Rightarrow \impliedCI$ then $\iset(\impliedCI) \subseteq \iset(\Sigma)$.
}
In the following Lemma, we apply Theorem~\ref{thm:YeungBuildMeasure} to characterize exact implication in the polymatroid cone $\entropicPlhdrl_n$.
\begin{lem}\label{thm:implicationInclusion}
	\implicationInclusion
\end{lem}
\begin{proof}
	Assume, by contradiction, that $\iset(\impliedCI) \not\subseteq \iset(\Sigma)$, and let $b\in \iset(\impliedCI){\setminus}\iset(\Sigma)$. By Theorem~\ref{thm:YeungBuildMeasure} the $I$-measure $\imeasure$ can take the following non-negative values: 
	\begin{equation*}
		\imeasure(a)=
		\begin{cases}
			1 & \text{if } a=b \\
			0 & \text{otherwise}
		\end{cases}
	\end{equation*}
	It is evident that $\imeasure(\Sigma)=0$ while $\imeasure(\impliedCI)= 1$, which contradicts the implication. \eat{ $\entropicPlhdrl_n \models_{EI} \Sigma \Rightarrow \impliedCI$.} 
\end{proof}

\subsection{Proof from Section~\ref{sec:PBoundedRelaxations}}
\label{sec:ProofFromSec4}
In this section we formally prove Theorem~\ref{thm:DisjointSaturated} from Section~\ref{sec:PBoundedRelaxations} that relies on the I-measure. For brevity, we denote the intersection of the sets corresponding to variables $A$ and $B$ (e.g., $\iset(A)\cap \iset(B)$) by $\iset(A)\iset(B)$.
\label{sec:IMeasureProofs}
\begin{reptheorem}{\ref{thm:DisjointSaturated}}
\DisjointSaturated
\end{reptheorem}
	\begin{proof}
	From Corollary~\ref{corr:PniffGamman}
	we have that $\entropicPlhdrl_n \models_{\EI} \Sigma \implies \tau$.
	Let $\impliedCI \eqdef (A;B|C)$ be the saturated implied CI, and let $\sigma=(X;Y|Z) \in \Sigma$ be a saturated CI. 
	Recalling that $\iset(\tau)=\isetc(C)\iset(A)\iset(B)$ (see Table~\ref{tab:ImeasureSummary})\footnote{And recall that $\iset(A)\cap\iset(B)$ is denoted $\iset(A)\iset(B)$} 
	we get that:
	\begin{align}
		\isetc(\impliedCI)&=\left(\isetc(C)\iset(A)\iset(B)\right)^c \nonumber \\
		&=(\isetc(C))^c \cup \isetc(A) \cup \isetc(B) \nonumber \\
		&=\iset(C) \cup \isetc(C)\isetc(A)\cup \isetc(C)\iset(A)\isetc(B) \label{eq:satdisjProof1}
	\end{align}
	Furthermore, since $\sigma=(X;Y|Z)$ is saturated, we get that $C=C_XC_YC_Z$ where $C_X=C\cap X$, $C_Y=C \cap Y$, and $C_Z=C\cap Z$. From this, we get that:
	\begin{equation}\label{eq:satdisjProof2}
		\iset(C)=\iset(C_X)\cup \isetc(C_X)\iset(C_Y) \cup \isetc(C_X)\isetc(C_Y)\iset(C_Z)
	\end{equation}
	Then, from~\eqref{eq:satdisjProof1}, and the set-additivity of $\imeasure$, we get that:
	\begin{align}
		&\imeasure\left(\iset(\sigma) \cap \isetc(\impliedCI)\right) \nonumber\\
		&=\imeasure(\iset(\sigma)\cap \iset(C))+\imeasure(\iset(\sigma)\cap \isetc(C)\isetc(A))+\imeasure(\iset(\sigma)\cap \isetc(C)\iset(A)\isetc(B)) \label{eq:satdisjProof3}
	\end{align}
	We consider each one of the terms in~\eqref{eq:satdisjProof3} separately. For the first term, by~\eqref{eq:satdisjProof2}, we have that:
	\begin{align}
		&\imeasure(\iset(\sigma)\cap \iset(C)) \nonumber\\
		&=\imeasure(\iset(\sigma) \iset(C_X))+\imeasure(\iset(\sigma) \isetc(C_X)\iset(C_Y))+\imeasure(\iset(\sigma) \isetc(C_X)\isetc(C_Y)\iset(C_Z)) \nonumber \\
		&=\imeasure(\iset(C_X)\iset(Y)\isetc(Z))+\imeasure(\iset(X{\setminus} C_X)\iset(C_Y)\isetc(ZC_X)) \nonumber \\
		&\qquad+\imeasure(\iset(X{\setminus} C_X)\iset(Y{\setminus} C_Y)\iset(C_Z)\isetc(ZC_XC_Y)) \label{eq:satdisjProofRev1}\\
		&=\imeasure(\iset(C_X)\iset(Y)\isetc(Z))+\imeasure(\iset(X{\setminus} C_X)\iset(C_Y)\isetc(ZC_X)) \label{eq:satdisjProof4}\\
		&=I_h(C_X;Y|Z)+I_h(X{\setminus} C_X;C_Y|ZC_X) \label{eq:satdisjProof5}
	\end{align}
	where transition~\eqref{eq:satdisjProofRev1} follows from substituting $\iset(\sigma)=\iset(X)\iset(Y)\isetc(Z)$, and
	transition~\eqref{eq:satdisjProof4}
	is because $C_Z \subseteq Z$ and thus $\iset(C_Z)\isetc(ZC_XC_Y)=\emptyset$, and $\imeasure(\emptyset)=0$.
	We now consider the second term of~\eqref{eq:satdisjProof3}.
	\begin{align}
		\imeasure(\iset(\sigma)\cap \isetc(C)\isetc(A)) 
		&=\imeasure(\iset(X)\iset(Y)\isetc(Z)\isetc(C)\isetc(A)) \nonumber \\
		&=\imeasure(\iset(X)\iset(Y)\isetc(ZCA)) \nonumber\\
		&=I_h(X;Y|ZCA)\label{eq:satdisjProof6}
	\end{align}
	Finally, we consider the third term of~\eqref{eq:satdisjProof3}. Since $\sigma=(X;Y|Z)$ is saturated, we get that $A=A_XA_YA_Z$ where $A_X=A\cap X$, $A_Y=A\cap Y$, and $A_Z=A\cap Z$. Therefore:
	\begin{equation}\label{eq:MVDDisjointProof}
		\iset(A)=\iset(A_X)\cup \isetc(A_X)\iset(A_Y)\cup \isetc(A_X)\isetc(A_Y)\iset(A_Z)
	\end{equation}
	From~\eqref{eq:MVDDisjointProof}, and the set-additivity of $\imeasure$, we get that:
	\begin{align}
		&\imeasure(\iset(\sigma)\cap \isetc(C)\iset(A)\isetc(B)) \nonumber\\ 
		&=\imeasure(\iset(X)\iset(Y)\isetc(Z)\isetc(C)\iset(A)\isetc(B)) \nonumber \nonumber\\
		&=\imeasure(\iset(X)\iset(Y)\iset(A)\isetc(ZCB)) \nonumber \\
		&=\imeasure(\iset(A_X)\iset(Y)\isetc(ZCB))+
		\imeasure(\iset(X{\setminus} A_X)\iset(A_Y)\isetc(ZCBA_X)) \nonumber \\
		&\qquad +\imeasure(\iset(X{\setminus} A_X)\iset(Y{\setminus}A_Y)\iset(A_Z)\isetc(ZCBA_XA_Y)) \label{eq:satdisjProofRev2}\\
		&=\imeasure(\iset(A_X)\iset(Y)\isetc(ZCB))+\imeasure(\iset(X{\setminus} A_X)\iset(A_Y)\isetc(ZCBA_X)) \label{eq:satdisjProof7}\\
		&=I_h(A_X;Y|ZCB)+I_h(X{\setminus} A_X;A_Y|ZCBA_X) \label{eq:satdisjProof8}
	\end{align}
	where transition~\eqref{eq:satdisjProofRev2} follows from substituting the formula for $\iset(A)$ (see~\eqref{eq:MVDDisjointProof}), and
	where transition~\eqref{eq:satdisjProof7} is because $A_Z \subseteq Z$, and thus $\iset(A_z)\cap \isetc(ZCBA_X)=\emptyset$, and $\imeasure(\emptyset)=0$.
	By~\eqref{eq:satdisjProof5}, \eqref{eq:satdisjProof6}, and~\eqref{eq:satdisjProof8} we get that:
	\begin{align}
		\imeasure(\iset(\sigma)\cap \isetc(\tau))&=I_h(C_X;Y|Z)+I_h(X{\setminus}C_X;C_Y|ZC_X) \nonumber \\
		&+I_h(X;Y|ZCA)+I_h(A_X;Y|ZCB)+I_h(X{\setminus}A_X;A_Y|ZCBA_X) \label{eq:satdisjProof9}
	\end{align}
	Therefore, by~\eqref{eq:satdisjProof9}, we get that: $\imeasure\left(\iset(\sigma) \cap \isetc(\impliedCI)\right) \geq 0$
	for every $\sigma \in \Sigma$.
	By Lemma~\ref{thm:implicationInclusion} we have that $\iset(\impliedCI) \subseteq \iset(\Sigma)=\bigcup_{\sigma\in \Sigma}\iset(\sigma)$. And since $\Sigma$ is pairwise disjoint
	then:
	\[
	\imeasure(\iset(\impliedCI))=\imeasure(\bigcup_{\sigma \in \Sigma}\iset(\sigma)\cap \iset(\tau))\underbrace{=}_{\Sigma \text{ is pairwise disjoint}}=\sum_{\sigma \in \Sigma}\imeasure(\iset(\impliedCI)\cap \iset(\sigma))
	\]	
	The result then follows from noting that due to the set-additivity of $\imeasure$ we have:
	$$h(\Sigma)=\sum_{\sigma \in \Sigma}\imeasure(\iset(\sigma))=\underbrace{\sum_{\sigma\in \Sigma}\imeasure(\iset(\sigma)\cap \iset(\tau))}_{\substack{=\imeasure(m(\tau))=h(\tau)\\ \text{ because }\Sigma \text{ is disjoint}}}+\underbrace{\sum_{\sigma\in \Sigma}\imeasure(\iset(\sigma)\cap \isetc(\tau))}_{\geq 0 \text{ (by~\eqref{eq:satdisjProof9})}}.\qedhere$$
\end{proof}

\added{
	\subsection{Exact implications in the cone $\positiveConen$ admit a 1-relaxation} \label{sec:Pn1Relaxations}
	We recall the signed measure $\imeasure$ (Theorem~\ref{thm:YeungUniqueness}) defined over the field $\mathcal{F}_n$ (Definition~\ref{def:field}) whose atoms we denote by $\mathcal{A}$. 
	In particular, $\imeasure$ can be negative.
	
		\begin{exa}
			\label{ex:NegativeIMeasure}
		Consider the parity function on three binary random variables. That is, $Z=X \oplus Y$ where $X,Y$ are uniformly distributed binary variables. We first observe that the variables are pairwise independent. Clearly, $X$ and $Y$ are pairwise independent. It is also easy to see that $P(Z=z)=\frac{1}{2}$ for $z\in \set{0,1}$. Observe that $X$ and $Z$ are also independent:		
		\begin{align*}			
			 P(X=x,Z=z)&=P(x,z,Y=z\oplus x)+P(x,z,Y\neq z\oplus x)\\
			 &=P(z|x,Y=z\oplus x)P(x,Y=z\oplus x)+P(z|x,Y\neq z\oplus x)P(x,Y\neq z\oplus x)\\
			 &=P(z|x,Y=z\oplus x)P(x,Y=z\oplus x)\\
			 &=1 \cdot P(x,Y=z\oplus x) \\
			 &=\frac{1}{4} \\
			 &=P(x)\cdot P(z)		 
		 \end{align*}	 
		 Further, every variable is functionally determined by the other two (e.g., $X=0,Z=0 \implies Y=0$). Let $h_P$ denote the entropic function associated with $P$. Since $X$ and $Y$ are independent, then $I_{h_P}(X,Y)=0$. On the other hand, given $Z=z$ variables $X,Y$ are not independent.
		 We can see this formally:
		 \begin{align}
		 	I_{h_P}(X;Y|Z)&=h_P(XZ)+h_P(YZ)-h_P(Z)-h_P(XYZ) \nonumber \\
		 	&=h_P(X)+h_P(Y)+h_P(Z)-(h_P(Z|XY)+h(XY)) \nonumber \\
		 	&=h_P(X)+h_P(Y)+h_P(Z)-h_P(XY) \nonumber \\
		 	&=h_P(Z) \nonumber \\
		 	&=2\cdot \frac{1}{2}\log 2= \log 2 > 0 \label{eq:I(X;Y|Z)>0}
		 \end{align}
		We now see how this is related to the i-measure $\imeasure$. By definition of $\imeasure$ we have that (see Table~\ref{tab:ImeasureSummary}):
		\begin{align*}
			0=I_{h_P}(X,Y)&=\imeasure(\iset(X)\cap \iset(Y)) \\
			&=\imeasure(\iset(X)\cap \iset(Y)\cap \iset(Z))+\imeasure(\iset(X)\cap \iset(Y)\cap \isetc(Z))		
		\end{align*}
	    Therefore,
	    \begin{align*}
			\imeasure(\iset(X)\cap \iset(Y)\cap \iset(Z))&=-\imeasure(\iset(X)\cap \iset(Y)\cap \isetc(Z))& \\
			&=-I_{h_P}(X;Y|Z)& \text{Table~\ref{tab:ImeasureSummary} } \\
			&= -\log 2& \eqref{eq:I(X;Y|Z)>0} \\
			&<0
		\end{align*}		
	\end{exa}

	If $\imeasure(a) \geq 0$ for all atoms $a \in \mathcal{A}$ then it is called a \e{positive measure}. A polymatroid is said to be \e{positive} if its I-measure is positive, and we denote by $\Delta_n$ the cone of positive polymatroids. 
	
	Theorem~\ref{thm:YeungBuildMeasure} states
	that any I-measure assigning non-negative values to its atoms is a polymatroid.
	As a corollary, we get, in Lemma~\ref{thm:implicationInclusion}, a set-theoretic characterization of exact implication in the polymatroid cone.
	The result of Lemma~\ref{thm:implicationInclusion}, combined with the assumption of a non-negative i-measure, gives us a 1-relaxation
	in the cone of positive polymatroids $\Delta_n$ (Theorem~\ref{thm:PositiveCone1Relaxation}). 	
}
	\begin{thm}
		\label{thm:PositiveCone1Relaxation}
		Exact implication in the cone of positive polymatroids admits a 1-relaxation.
		\[
		\text{If } \Delta_n \models_{\EI} \Sigma \Rightarrow \tau \text{ then } \Delta_n \models h(\Sigma) \geq h(\tau).
		\]	
	\end{thm}
	\begin{proof}
		By Lemma~\ref{thm:implicationInclusion} we have that $\iset(\impliedCI) \subseteq \iset(\Sigma)$. 
		Hence, $\iset(\impliedCI)= \cup_{\sigma \in \Sigma}\left(\iset(\impliedCI)\cap \iset(\sigma)\right)$. By the assumption that the I-measure $\imeasure$ is non-negative, we get that:
		\begin{equation}\label{eq:imeasureOfCInPositiveCone}
			\imeasure(\iset(\impliedCI))= \imeasure\left( \cup_{\sigma \in \Sigma}\left(\iset(\impliedCI)\cap \iset(\sigma)\right)\right) \leq \sum_{\sigma \in \Sigma}\imeasure((\iset(\impliedCI)\cap \iset(\sigma))
		\end{equation}
		By the set-additivity of $\imeasure$ we get that:
		
		\begin{equation}\label{eq:inclusionInSigma}
			h(\Sigma)\eqdef\sum_{\sigma\in \Sigma}\imeasure(\iset(\sigma))=\sum_{\sigma\in \Sigma}\imeasure(\iset(\sigma)\cap \iset(\impliedCI))+\sum_{\sigma\in \Sigma}\imeasure(\iset(\sigma)\cap \isetc(\impliedCI))
		\end{equation}
		From~\eqref{eq:inclusionInSigma},~\eqref{eq:imeasureOfCInPositiveCone}, and non-negativity of $\imeasure$ we get that:
		\begin{align*}
			\imeasure(\iset(\impliedCI))&\leq \sum_{\sigma \in \Sigma}\imeasure(\iset(\sigma))-\sum_{\sigma\in \Sigma}\imeasure(\iset(\sigma)\cap \isetc(\impliedCI))\\
			&\leq \sum_{\sigma\in \Sigma}\imeasure(\iset(\sigma)).
		\end{align*}
		Since, by Theorem~\ref{thm:YeungUniqueness}, $\imeasure$ is the unique signed measure consistent with all Shannon information measures, the result follows.
	\end{proof}

	\paragraph{\textbf{$\positiveConen$ coincides with the cone of step functions}}\label{sec:IMeasureStepFunctions}
	\eat{
		\dan{We should condense this section to a single theorem, namely
			Th.~\ref{th:step:positive}, which I would restate by saying this:
			``$\Delta_n$ consists precisely of the entropic functions with a
			non-negatie I-measure''.  We don't need an entire section for this.
			The proof can go to the appendix.}
		
		Then, in Theorem~\ref{th:step:positive} we prove that the conic hull of step functions $\positiveConen=\conhull(\stepfn)$ coincides with the cone of positive polymatroids $\Delta_n$ (i.e., $\positiveConen=\Delta_n$).
		In particular, since $\positiveConen \subseteq \Delta_n$, we get a 1-relaxation in $\positiveConen$.
	}
	
	We describe the structure of the I-measure of a step function, and prove that the conic hull of step functions $\positiveConen=\conehull{\stepfn}$ and positive polymatroids $\Delta_n\subseteq \Gamma_n$ coincide. That is, $\Delta_n=\positiveConen$.
	We let $U \subseteq [n]$, and let $s_U$ denote the step function at $U$.
	In the rest of this section we prove Theorem~\ref{th:step:positive}.
	
	\def\PositiveAndStepConesEqual{
		It holds that $\Delta_n=\positiveConen$.
	}
	\begin{thm} \label{th:step:positive}
		\PositiveAndStepConesEqual
	\end{thm}
	For proving this Theorem, we require Lemma~\ref{lemma:StepFunctionThm} that characterizes the i-measure function for step functions. In particular, this lemma shows that step functions are positive polymatroids (i.e., have positive i-measures).
	\def\StepFunctionThm{
		Let $s_U$ be the step function at $U \subseteq [n]$, and $R=[n]\setminus U$.
		The unique I-measure for any step function $s_U$ is:
		\begin{equation}\label{eq:StepFunctionThm}
			\imeasure(a)=
			\begin{cases}
				1 & \text{if }a=\left(\cap_{X\in R}\iset(X)\right)\cap \left(\cap_{X\in U}\isetc(X)\right)\\
				0 & \text{otherwise}
			\end{cases}
		\end{equation}
	}	
	\begin{lem}\label{lemma:StepFunctionThm}
		\StepFunctionThm
	\end{lem}
	\begin{proof}
		We define the atom $a^*$ by:
		\[
		a^*\eqdef\left(\cap_{X\in R}\iset(X)\right)\cap \left(\cap_{X\in U}\isetc(X)\right)
		\]
		Let $S \subseteq U$. We show that $h(S)=\imeasure(\iset(S))=0$.
		By definition, $\iset(S)=\cup_{Y\in S}\iset(Y)$. Therefore, every atom $a\in \iset(S)$ has at least one set $\iset(Y)$, where $Y\in S$, that appears in positive form.
		In particular, $a^*\notin \iset(S)$. Therefore, by~\eqref{eq:StepFunctionThm}, we have that $h(S)=\imeasure(\iset(S))=0$ as required.
		
		Now, let $T{\not\subseteq} U$, and let $X {\in} T{\setminus} U$. 
		Since $\iset(X)$ appears in positive form in $a^*$, and since $\iset(X) {\subseteq} \iset(T)$, then we have that $h(T){=}1$.
		Therefore, $\imeasure$ as defined in~\eqref{eq:StepFunctionThm} is the I-measure for $s_U$.
		By Theorem~\ref{thm:YeungUniqueness}, the I-measure $\imeasure$ is \e{unique}, and is therefore the only I-measure for $s_U$.
	\end{proof}
	
	\def\PositiveAndStepConesEqual{
		It holds that $\Delta_n=\positiveConen$.
	}
	By Lemma~\ref{lemma:StepFunctionThm}, every step function has a non-negative
	I-measure. As a consequence, any conic combination of step functions has a positive I-measure, proving that $\positiveConen \subseteq \Delta_n$.
	
	Now, let $\imeasure$ denote any non-negative I-measure.
	For every atom $a$, we let $\vneg(a)$ denote the set of variables whose sets appear in negated form (e.g., if $a=\left(\cap_{X\in R}\iset(X)\right)\cap \left(\cap_{X\in U}\isetc(X)\right)$, then $\vneg(a)=U$).
	Let $\mathcal{A}$ denote the atoms of $\mathcal{F}_n$.
	We define the mapping $f: \mathcal{A} \mapsto \positiveConen$ as follows.
	\[
	f(a)\eqdef\imeasure(a)\cdot s_{\vneg(a)}
	\]
	where $s_{\vneg(a)}$ is the step function at $\vneg(a)$.
	Since, by Lemma~\ref{lemma:StepFunctionThm}, the I-measure of $s_{\vneg(a)}$ is $1$ at atom $a$ and $0$ everywhere else we get that:
	\[
	\imeasure=\sum_{a\in \mathcal{A}}f(a)=\sum_{a\in \mathcal{A}}\imeasure(a)\cdot s_{\vneg(a)}
	\]
	proving that $\imeasure \in \positiveConen$, and that $\Delta_n \subseteq \positiveConen$.
This concludes the proof of Theorem~\ref{th:step:positive}.
\section{Discussion and Future Work}

\paragraph{Number of Repairs.}
A natural way to measure the degree of a
constraint in a relation instance $R$ is by the number of repairs
needed to enforce the constraint on $R$.  In the case of a key
constraint, $X \fd Y$, where $XY=\Omega$, our information-theoretic
measure is naturally related to the number of repairs, as follows.  If
$h(Y|X) = c$, where $h$ is the entropy of the empirical distribution
on $R$, then one can check $|R|/|\Pi_X(R)| \leq 2^c$.  Thus, the
number of repairs $|R|-|\Pi_X(R)|$ is at most $(2^c-1)|\Pi_X(R)|$.  We
leave for future work an exploration of the connections between the number
of repairs and information theoretic measures.

\paragraph{Small Model Property.} We have proven in
Sec.~\ref{sec:PBoundedRelaxations} that several classes of
implications (including saturated CIs, FDs, and MVDs) have a ``small
model'' property: if the implication holds for all uniform, 2-tuple
distributions, then it holds in general.  In other words, it suffices
to check the implication on the step functions $\stepfn$.  One
question is whether this small model property continues to hold for
other tractable classes of implications in the literature.  For
example, Geiger and Pearl~\cite{GeigerPearl1993} give an
axiomatization (and, hence, a decision procedure) for \e{marginal}
CIs.  However, marginal CIs do not have the same small model property.
Indeed, the implication
$(X \perp Y), (X \perp Z) \Rightarrow (X \perp YZ)$
holds for all
uniform 2-tuple distributions (because
$I_h(X;YZ) \leq I_h(X;Y) + I_h(X;Z)$ holds for all step functions),
however it fails for the ``parity distribution'' in
Fig.\ref{fig:examples}(b).  We leave for future work an investigation
of the small model property for other classes of constraints.

\paragraph{Bounds on the Factor $\lambda$.} In the early stages of this work
we conjectured that all CIs in $\Gamma_n$ admit a 1-relaxation, until
we discovered the counterexample in Th.~\ref{thm:NoUAI}, where
$\lambda = 3$.  On the other hand, the only general upper bound is
$(2^n)!$.  None of them is likely to be tight.  We leave for future
work the task of finding tighter bounds for $\lambda$.

\bibliographystyle{alphaurl} 
\bibliography{main}

\end{document}